%% file: main.tex
\documentclass[preprint,11pt,3p,authoryear]{elsarticle}

\usepackage{amssymb}
\usepackage[rightcaption]{sidecap}

\usepackage{todonotes}
\definecolor{myblue}{rgb}{0.8,0.8,1}
\definecolor{myred}{rgb}{1,0.8,0.8}
\definecolor{mygreen}{rgb}{0.8,1,0.8}

\usepackage{subfigure}
\journal{TBA}
\usepackage{hyperref} 
\hypersetup{ colorlinks, citecolor=black, filecolor=black, linkcolor=black, urlcolor=black } 
\usepackage{float,nicefrac,amsmath}
\newcommand{\diff}{\mathrm{d}}

\usepackage{amsmath, amsthm, bm}
\usepackage{color,soul}
\usepackage{bbm}
\newcommand\norm[1]{\left\lVert#1\right\rVert}

\newtheorem{Assumption}{Assumption}
\newtheorem{theorem}{Theorem}

\usepackage{titlesec}
\titleformat*{\section}{\Large\bfseries}
\titleformat*{\subsection}{\large\bfseries}
\titleformat*{\subsubsection}{\bfseries}
\usepackage{enumitem}

\DeclareMathOperator*{\argmin}{argmin}

\usepackage{bigstrut}

\newcommand{\rmes}{p}
\newcommand{\rcompr}{\tilde{p}}
\newcommand{\rcompd}{q}

\newcommand{\mfT}{{\mathfrak{T}}}

\newcommand{\mfO}{{\mathfrak{O}}}

\newcommand{\tT}{{t\in\mfT}}
\newcommand{\mcA}{{\mathcal{A}}}
\newcommand{\mcC}{{\mathcal{C}}}
\newcommand{\mcF}{{\mathcal{F}}}
\newcommand{\Ffct}{{\hat{G}}}

\newtheorem{proposition}{Proposition}
\newtheorem{lemma}{Lemma}
\newtheorem{corollary}{Corollary}
\newcommand\abs[1]{\left|#1\right|}

\reversemarginpar

\begin{document}
\begin{frontmatter}

\title{\textbf{Latency and Liquidity Risk}}
\tnotetext[label0]{SJ would like to acknowledge support from the Natural Sciences and Engineering Research Council of Canada (grants RGPIN-2018-05705 and RGPAS-2018-522715). LSB acknowledges support from CONACyT, M\'exico and the Mathematical Institute, University of Oxford.}

\author[label1,label2]{\'{A}lvaro Cartea}

\address[label1]{Mathematical Institute, University of Oxford}
\address[label2]{Oxford-Man Institute of Quantitative Finance}
\ead{alvaro.cartea@maths.ox.ac.uk}

\address[label4]{Department of Statistical Sciences, University of Toronto}

\author[label4]{Sebastian Jaimungal}
\ead{sebastian.jaimungal@utoronto.ca}

\author[label1]{Leandro S\'{a}nchez-Betancourt}
\ead{leandro.sanchezbetancourt@maths.ox.ac.uk}

\begin{abstract}
Latency (i.e., time delay) in electronic markets affects the efficacy of liquidity taking strategies. During the time liquidity takers process information and send marketable limit orders (MLOs) to the exchange, the limit order book (LOB) might undergo updates, so there is no guarantee that MLOs are filled. We develop a latency-optimal trading strategy that improves the marksmanship of liquidity takers. The interaction between the LOB and MLOs is modelled as a marked point process. Each MLO specifies a price limit so the order can receive worse prices and quantities than those the liquidity taker targets if the updates in the  LOB are against the interest of the trader. In our model,   the liquidity taker balances the tradeoff between missing trades and the costs of walking the book. We employ techniques of variational analysis to obtain the optimal price limit of each MLO the agent sends. The price limit of a MLO  is characterized as the solution to a new class of forward-backward stochastic differential equations (FBSDEs) driven by random measures. We prove the existence and uniqueness of the solution to the FBSDE and numerically solve it to illustrate the performance of the latency-optimal strategies. 
\end{abstract}

\begin{keyword}
Marked point processes \sep high-frequency trading \sep algorithmic trading \sep latency \sep forward-backward stochastic differential equations.
\end{keyword}

\end{frontmatter}

\section{Introduction}
 
Speed to make decisions and to access the market is a  key element  in the success of trading strategies in electronic markets. Liquidity providers monitor and update their limit orders (LOs) resting in the limit order book (LOB),  and liquidity takers send orders that target LOs. The efficacy of the strategies of the makers and takers of liquidity depends on their latency in the marketplace. Latency is the time delay between an exchange streaming market data to a trader, the trader processing information and making a decision,  and the exchange receiving the instruction from the trader. Thus, due to latency, there is no guarantee that liquidity providers can place a LO in a desired queue position in the book or withdraw a stale quote before it is picked off by another trader. 

Furthermore, there are  no assurances that marketable limit orders (MLOs)  from liquidity takers, which aim at a quantity and price they observed in the LOB, hit the desired target.  A MLO is a liquidity taking order for immediate execution against the LOs resting in the book, and each  MLO specifies the quantity of the security (e.g., equity, currency pairs, futures, etc.) and a price limit to execute against LOs.\footnote{A marketable order and a market order differ in that the marketable limit order walks the LOB until it reaches the limit price specified by the trader, while a market order walks the LOB until it is filled in full.} Due to latency,  by the time the exchange processes a MLO, prices and quantities  could have  improved,  so the order is filled at a better price, or prices and quantities could have worsened, so the order is filled if the limit price allows, otherwise the order is rejected.

In this paper, we focus on how latency affects the marksmanship of liquidity takers and we develop a latency-optimal trading strategy that accounts for the time delays in the marketplace.  We frame the interaction between the LOB and MLOs as a marked point process (MPP). In our model,  the agent sends buy/sell MLOs at random times to partly camouflage her order flow, and before the order reaches the exchange, the LOB undergoes quantity and price updates.  We assume the agent sends fill-or-kill MLOs, that is, the orders are either filled in full or  rejected.\footnote{This is in contrast to an immediate-or-cancel order, which has the property that the order can be partially filled if there is liquidity in the LOB that meets the requirements of the MLO. The unfilled portion of the order is rejected. }  The price limit of the MLO consists of the best quote the agent observes at the time she decides to trade and a discretion to walk the LOB.

The LOB is a moving target, so liquidity takers hit or miss the LOs they are attempting to execute. Everything else being the same, the chances of filling a MLO increase if the agent is willing to receive quantities and prices that are worse than those of the best quotes the agent observes in the LOB when she decides to trade.   If the discretion to walk the book is unlimited, the MLO will be filled, but potentially at much worse prices than those of the best quotes the agent observed. On the other hand, if the updates in the LOB are in the  interest of the agent, the MLO will be filled at  better prices than those of the LOs that the agent targeted.

In our model, the agent balances the tradeoff between missing trades   and the costs from walking the LOB over a trading window (e.g., minutes, hours, days, etc.). For each liquidity taking order, the strategy   optimizes the discretion  of the MLO, while it penalizes both the number of missed trades and the costs accrued to the strategy over the trading horizon.   We employ techniques of variational analysis to obtain the optimal discretion  for each MLO the agent sends, which we  characterize  as the solution to a forward-backward stochastic differential equation (FBSDE).  We show existence and uniqueness of the solution to the forward and backward parts of the FBSDE and show existence and uniqueness of the solution to the full FBSDE. To  the best of our knowledge, uniqueness and existence of the resulting random-measure driven FBSDE is not covered  in the extant literature, and the particular form itself appears to be new.

In the agent's performance criterion, when the penalty for missing trades is linear in the expected number of rejected trades, we obtain the optimal strategy in closed-form -- the latency-optimal strategy consists of sending all MLOs with a fixed discretion. When the penalty for missing trades is quadratic in the expected  number of rejected trades, we solve the  FBSDE numerically. We illustrate the performance of the latency-optimal strategies for a range of model parameters and examine the tradeoff between costs from walking the book and number of missed trades. Finally, we discuss strategies that are cost-neutral to the agent. That is, the latency-optimal strategy is devised so the expected costs from walking the book to fill MLOs when the LOB moves against the agent's interests is the same as the expected benefits (i.e., negative costs) from executing trades at better prices than the ones the agent targets.

 
Several authors  address various aspects of latency in electronic markets.   \cite{MoallemiSaglam13} look at the cost of latency for liquidity takers in equity markets.  They compare the costs of  liquidating one stock with and without time delays in the marketplace to  compute the cost of latency. The work of \cite{SashaWaeber16} shows how to execute a large  order in electronic markets by employing the volume imbalance of the LOB to predict price changes and study the effect of latency in the efficacy of the execution strategy.   \cite{doi:10.1142/S2382626617500095} employ data from Nasdaq-Omx  and also find that as latency increases, the informational content in the volumes of the LOB diminishes.

\cite{CarteaLeo2018} employ proprietary foreign exchange data to show how latency and  volatility of the midprice of the security affect the fill ratio of liquidity taking strategies. The authors show how traders could employ latency-optimal strategies to improve fill ratios, while minimizing costs, and they show how to compute the  shadow price of latency in foreign exchange markets.   \cite{gao2018electronic}  use Markov decision processes to model the problem of a market maker with latency who trades in a LOB, where the size of the  quoted spread  is always one tick.  The authors find that as latency increases, the profits from making markets decrease.

Recent  literature on high-frequency trading and algorithmic trading discusses various characteristics of trading and how traders use speed  to obtain informational advantages, see e.g., \cite{sophie2018market}. Other strands of the literature discuss the relationship of market quality, the speed of market participants, and stochastic liquidity,  see for example  \cite{almgrensiam} and \cite{gueant2016financial} for trading in illiquid markets. \cite{barger2018optimal} model the rapid updates of the best quotes in the LOB to propose a model of stochastic price impact.
 
The remainder of the paper proceeds as follows. Section \ref{sec:optimal discretion to walk book} proposes the agent's performance criterion and characterizes the latency-optimal strategy as the solution to a FBSDE. Section \ref{sec:existence and uniqueness}  shows existence and uniqueness of the solution of the forward and backward part of the FBSDE, and existence and uniqueness of the solution to the full FBSDE.  Section  \ref{sec:optimality} shows that the candidate control we find is the global optimum and Section \ref{sec:numerical results}  discusses the performance of the strategy for various scenarios. We conclude in Section \ref{sec:conclusions} and collect some proofs in the Appendix.
 
\section{Optimal discretion to walk the book }\label{sec:optimal discretion to walk book}

\subsection{Latency: the LOB as a  moving target}
Liquidity takers in electronic markets face a moving target problem. Traders send orders that target a price and quantity  they observe in the LOB, but due to latency, when the order arrives in the exchange, the target could have moved. If prices and quantities worsen,  the agent's order is rejected, and  if prices and quantities  improve or do not worsen, the order is filled.

We frame the moving target problem as a MPP  $\mathcal{N}=\{(T_n,\,Z_n)\}_{(n\geq1)}$  in the probability space $\left(\Omega,\,\mathcal{F},\,\mathbb{P}\right)$. Here, $(T_n)$ is an increasing sequence of random points in $(0,T]\bigcup \{\infty\}$, which represent the times when the agent sends MLOs to the exchange, and $(Z_n)$  is a sequence of marks, which  represent the shock to the average price per share due to changes in  prices and quantities. 

We assume that each order is for one unit of the security or for a lot of securities, where the lots have a fixed size throughout the trading horizon. When the volume of the MLO is in lots of the security, the mark $Z$ represents a shock to the LOB commensurate with the volume of the MLO.

As in \cite{confortola2016}, we define the sample space to be  $\Omega={\bigcup}_n \{T_n>T\}$, where $T\in (0,\infty)$ denotes a fixed time horizon. The filtration $\left(\mathcal{F}_t\right)_{t\geq 0}$ is  generated by $\mathcal N$ and is the smallest filtration such that for each $n$, the point $T_n$ is a stopping time and the mark $Z_n$ is $\mathcal{F}_{T_n}$-measurable. We use predictable processes to mean the left-continuous version of a process, see Theorem 7.2.4 in \cite{Cohen15}.

The random measure  associated with  $\mathcal{N}$ is 
\begin{equation*}
\rmes(\diff t,\diff z)=\sum_{n\geq 1: T_n\leq T}\mathfrak{D}_{(T_n,Z_n)}(\diff t,\diff z)\,,
\end{equation*}
where $\mathfrak{D}$  denotes the Dirac measure, and we assume that
\begin{equation}\label{Assumption_S.I.}
\mathbb{E}\left[\left(\rmes\left([0,T],\,\mathbb{R}\right)\right)^2\right]<\infty\qquad\text{and}\qquad \mathbb{E}\left[\int_0^T\int_{\mathbb{R}}\abs{z}\,p(\diff z,\diff t)\right]<\infty\,.
\end{equation}

We denote by $\rcompr$ the predictable compensator of the random measure $\rmes$, which admits the following decomposition  
\begin{equation}\label{compensator}
\rcompr(\diff z,\diff t)=\phi_t(\diff z)\,\diff A_t\,.
\end{equation}
Here, the compensator has the property that for $\rcompd:=\rmes-\rcompr$ and any integrable and predictable process $H$, the stochastic integral  $\left(H\star \rcompd\right)_t=\int_0^t\int_{\mathbb{R}}H_s\, \rcompd(\diff s,\diff z)$ is a martingale. In \eqref{compensator}, the predictable process  $(A_t)_{t\in\mfT}$, where $\mfT:=[0,T]$, is the compensator of the counting process of the MLOs, which we denote by $N_t:=\rmes([0,t]\times \mathbb{R})$.
\begin{Assumption}
The process $A$  admits a bounded stochastic intensity so that we may write $A_t=\int_0^t \lambda_u\,\diff u$ for a predictable process $\left(\lambda_t\right)_{\tT}$ and $\exists\,\, \bar{\lambda}\in\mathbb{R}$, such that  $\forall\, (t,\omega)\in[0,T]\times\Omega$,  $\lambda_t(\omega)\leq \bar{\lambda}$.
\end{Assumption}

The density function of the marks is $\phi_t$, which has support in $\mathbb{R}$ and is bounded, and its cumulative distribution function is $\Phi_t$, which  we assume  is uniformly Lipchitz in $[0,T]\times \Omega$ with Lipschitz constant $k$.

Let $\left(\delta_t\right)_{\tT}$ be a predictable process that specifies the  cash per unit of the security (or lots of the security) the agent is willing to  walk the LOB to increase the chances of filling her liquidity taking order, i.e., $\delta$ is the discretion of the MLO. For example, in equity markets,  if the agent sends a buy  order  to lift the offer at the best ask $a_t$, the discretionary amount $\delta_t$ is the extra cash per share the order may walk the book, i.e., $a_t+\delta_t$ is the highest price the agent is willing to pay for one share of equity. 
Similarly, if the agent sends a sell order to hit the best bid $b_t$, the amount $\delta_t$ is the cash discount per share the order may walk the book, i.e.,  $b_t-\delta_t$ is the lowest price the agent is willing to accept to sell one share of equity.

In the examples above, the best bid and best ask prices ($b_t$ and $a_t$) refer to those the agent `observes' when she decides to trade, but due to latency, these prices could be stale. In addition, by the time the exchange processes the order of the agent, prices and quantities in the LOB could have borne further updates. Price changes could be against or in favour of the agent's interest.  When the price per unit of the security moves against the interest of the agent, the order is filled only if the discretion $\delta$ of the MLO is enough to cover the adverse change in price and quantity; we refer to this as a price deterioration. On the other hand, if the price per unit of the security moves in favour of the agent's trade interest, the order is filled at a better price; we refer to this as a price improvement

\subsection{Tradeoff: cost of walking the LOB and missed trades}

The agent must balance the costs of walking the LOB against the number of missed trades as a consequence of her latency in the marketplace. Clearly, if the agent sends orders with infinite discretion to walk the LOB, all orders are filled (we rule out cases in which the LOB is empty) and the costs accrued from walking the LOB are expected to be highest. On the other hand, everything else being equal, lowering discretion, lowers the strategy's cost but increases the number of missed trades.

We discuss the cost for MLOs with volume equal to one unit of the security -- the costs for MLOs where volume is in lots of the security are computed in a similar way. For buy orders, the cost of the strategy is the cash the agent pays for the security minus the price on the offer side of the LOB that the agent targets. Similarly, for sell orders, the cost of the strategy is the target price in the bid side of the LOB minus the cash received for the security. That is, the cost of the strategy is the extra cash paid to walk the LOB, which is zero if the order is not executed. We denote the controlled cost process by $C^\delta= (C^\delta)_{\tT}$ and  
\begin{equation}\label{eqn:cost}
C^{\delta}_t=\int_0^t \int_{\mathbb{R}} z\,\Ffct(\delta_s-z)\,\rmes(\diff z,\diff s)\,,
\end{equation} 
where $\Ffct(x)=1$ if $x\geq 0$ and $\Ffct(x)=0$ otherwise.  

The extra cost  for each filled trade is $z\,\Ffct(\delta_s-z)$, which  can be negative (price improvement),   positive (price deterioration), or zero. This cost is negative when the shock to the LOB is negative ($z<0$), in which case  the order is filled at a better price than that targeted by the agent -- the price improvement is $\abs{z}$. On the other hand, this cost  is positive when the shock to the LOB is positive ($z>0$), in which case the order is filled (because $\delta\geq z$) at a worse price than that targeted by the agent -- the price deterioration is $z$. Finally, when the shock to the LOB is zero ($z=0$) or the trade is missed, the  cost is zero.

The process  $D^\delta= (D^\delta_t)_{\tT}$ denotes the controlled number of misses and  
\begin{equation}
D^{\delta}_t= \int_0^t \int_{\mathbb{R}}G\left(\delta_s-z\right)\,\rmes(\diff z,\diff s)\,,
\end{equation}
where $G(x)=1-\Ffct(x)$.  Recall that the MLO is for one unit of the security or for lots of the security, which are of fixed size throughout the trading horizon. In the latter case, the number of misses is in lots of the security.

 \subsection{Performance criterion}

The agent's performance criterion  is 
\begin{equation}\label{Functional}
J(\delta)=\mathbb{E}\left[C^{\delta}_T +\alpha\,D^{\delta}_T+\gamma\,\left(D^{\delta}_T\right)^2\right]\,,
\end{equation}
where both $\alpha\geq 0$  and $\gamma\geq 0$ are penalty parameters for the total number of missed trades, and the set of admissible strategies is
\begin{equation}\label{eqn:set of admissible strats}
    \mcA:=\left\{\delta = \left(\delta_t\right)_{\tT}\left| \delta\text{ is }\mathcal{F}-\text{{predictable}  and }
\mathbb{E}\left[\sup_{\tT}\left(\delta_t\right)^2 \right]\,<\,\infty \right.\right\}\,.
\end{equation}
The agent wishes to find a control $\delta^*\in\mcA$ that minimizes the performance criterion \eqref{Functional}, that is, the agent solves the problem 
\begin{equation*}
\delta^*=\argmin_{\delta\in\mcA}J(\delta)\,.
\end{equation*}
Note that  $J(\delta)<\infty$ because $G\leq 1$ and \eqref{Assumption_S.I.} holds. We choose the  units of the parameters $\alpha,\,\gamma$, so that the performance criterion has the same units as those of the costs $C$.

In the performance criterion, the penalties for missing trades are not financial costs. Everything else being equal, an increase in the value of the penalty parameters  makes the strategy post orders with higher discretion to walk the LOB. In the extreme case where  one of the penalty parameters is arbitrarily large, the optimal strategy is to post orders with discretion to walk the LOB as deep as necessary to fill the trades, i.e., the MLO with infinite discretion is a market order.

\subsection{Variational Analysis Approach}

We employ techniques of variational analysis to obtain the optimal discretion strategy.  For ease of presentation, we write 
\begin{equation}\label{eqn:performance criterion short}
J(\delta)=J^{\text{C}}(\delta)+\alpha\,J^{\text{LP}}(\delta)+\gamma\,J^{\text{QP}}(\delta)\,,
\end{equation}
where $J^{\text{C}}(\delta)=\mathbb{E}\left[C^{\delta}_T\right]$,  $J^{\text{LP}}(\delta)=\mathbb{E}\left[D^{\delta}_T \right]$, and  $J^{\text{QP}}(\delta)=\mathbb{E}\left[\left(D^{\delta}_T\right)^2 \right]$.

Next, note that
\begin{align}\label{eqn:jc}
J^{\text{C}}(\delta)&=\mathbb{E}\left[\int_0^T \int_{\mathbb{R}} z\,\Ffct(\delta_t-z)\,\rmes(\diff z,\diff s)\right]
=\mathbb{E}\left[\int_0^T \int_{\mathbb{R}} z\,\Ffct(\delta_t-z)\,\phi_t(\diff z)\,\diff A_t\right]\,,
\end{align}
and the next proposition   provides expressions for  $J^{LP}(\delta)$ and $J^{QP}(\delta)$.

\begin{proposition}
The following equations hold
\begin{align}\label{eqn:jlp}
J^{\text{LP}}(\delta)&=\mathbb{E}\left[\int_0^T\int_{\mathbb{R}}G(\delta_t-z)\,\rcompr(\diff z,\,\diff t)\right]\,,\\ \label{eqn:lqp}
J^{\text{QP}}(\delta)&=\mathbb{E}\left[\int_0^T\int_{\mathbb{R}}\left(2\,D^{\delta}_{t^-}\,G(\delta_t-z)+G(\delta_t-z)\right)\,\rcompr(\diff z,\,\diff t)\right]\,.
\end{align}
\end{proposition}
\begin{proof}
Equation \eqref{eqn:jlp} follows from the predictability of the integrand. Next, we show   \eqref{eqn:lqp}. The  number of missed trades $D^{\delta}_t$ satisfy the SDE
\begin{equation*}
\diff D^{\delta}_t=\int_{\mathbb{R}}G(\delta_t-z)\,\rmes(\diff z,\,\diff t)\,.
\end{equation*}
Let $h(x)=x^2$ and use an integration formula (see \cite{Jeanblanc2009}) to write
\begin{align*}
\diff h(D^{\delta}_t)=\int_{\mathbb{R}} \left(h\left(D^{\delta}_{t^-}+G\left(\delta_t-z\right)\right)-h\left(D^{\delta}_{t^-}\right)\right)\,\rmes(\diff z,\,\diff t)\,.
\end{align*}
Then,
\begin{align*}
\diff\left(D^{\delta}_t\right)^2&=\int_{\mathbb{R}}\left(2\,D^{\delta}_{t^-}\,G(\delta_t-z)+G^2(\delta_t-z)\right)\,\rmes(\diff z,\,\diff t)\\
&=\int_{\mathbb{R}}\left(2\,D^{\delta}_{t^-}\,G(\delta_t-z)+G(\delta_t-z)\right)\,\rmes(\diff z,\,\diff t)\,,
\end{align*}
where the second equality holds because $G^2=G$. Integrate from zero to $T$, take expectations, and because the integrand $2\,D^{\delta}_{t^-}\,G(\delta_t-z)+G(\delta_t-z)$ is predictable, obtain
\begin{align*}
\mathbb{E}\left[\left(D^{\delta}_T\right)^2\right]&=\mathbb{E}\left[\int_0^T\int_{\mathbb{R}}\left(2\,D^{\delta}_{t^-}\,G(\delta_t-z)+G(\delta_t-z)\right)\,\rmes(\diff z,\,\diff t)\right]\\
&=\mathbb{E}\left[\int_0^T\int_{\mathbb{R}}\left(2\,D^{\delta}_{t^-}\,G(\delta_t-z)+G(\delta_t-z)\right)\,\rcompr(\diff z,\,\diff t)\right]\,.
\end{align*}
\end{proof}

\subsubsection{Optimal discretion to walk the LOB}
We employ G\^ateaux derivatives to obtain the latency-optimal strategy  that minimizes the performance criterion of the agent.  Let $w,\delta\,\in\,\mcA$. The directional derivative of $J$ at $\delta$ in the direction of $w$ is given by
\begin{equation}
\langle \mathcal{D}\,J(\delta),w\rangle=\lim_{\epsilon\to 0}\tfrac{1}{\epsilon}\left[J(\delta+\epsilon \,w)-J(\delta)\right]\,,
\end{equation}
when the limit exists. Now, let $\mcA'$ be the dual space of $\mcA$. If there is $A'\in \mcA'$ such that $\langle \mathcal{D}\,J(\delta),w\rangle=A'(w)$ for all $w\in \mcA$, then $A'$ is called the G\^ateaux derivative of $J$ at $\delta$. In this paper, the directional derivatives  are elements of the dual of $\mcA$, hence we refer to the directional derivatives as G\^ateaux derivatives. Note that it is trivial to show that $\mcA$ is a linear space over $\mathbb{R}$.

\begin{lemma}\label{Gat:CostLPandQP}
The G\^ateaux derivative at $\delta\in\mcA$ in the direction $w\in\mcA$  of the:
\begin{enumerate}[label=(\alph*)]
\item cost functional $J^{\text{C}}$  is
\begin{align*}
\langle \mathcal{D}\,J^{\text{C}}(\delta),w\rangle&=\mathbb{E}\left[\int_0^T  w_t\,\phi_t(\delta_t)\,\delta_t\,\diff A_t\right]\,;
\end{align*}
\item linear penalty functional $J^{\text{LP}}$ is 
\begin{align*}
\langle \mathcal{D}\,J^{\text{LP}}(\delta),w\rangle&=-\mathbb{E}\left[\int_0^T w_t\,\phi_t(\delta_t) \,\diff A_t\right]\,;
\end{align*}
\item quadratic penalty functional $J^{\text{QP}}$ is 
\begin{align*}
\langle \mathcal{D}\,J^{\text{QP}}(\delta),w\rangle&=-2\,\mathbb{E}\left[\int_0^T w_t\,\phi_t(\delta_t)\,\mathbb{E}_{t^-}\left[\int_t^T\int_{\mathbb{R}}G(\delta_s-z') \, \rcompr(\diff z',\,\diff s)\right]  \,\diff A_t\right]\\
&\quad-2\,\mathbb{E}\left[\int_0^T w_t\,\phi_t(\delta_t)\,D^{\delta}_{t^-}\, \diff A_t\right]-\mathbb{E}\left[\int_0^T w_t\,\phi_t(\delta_t)\, \diff A_t\right]\,.
\end{align*}
\end{enumerate}
\end{lemma}
\begin{proof}
See \ref{proof of lemma}.
\end{proof}

The next theorem provides the G\^ateaux derivative of the performance criterion of the agent and provides a characterization of the optimal discretion to walk the LOB.

\begin{theorem}\label{Gat:J}
The G\^ateaux derivative of the functional $J$ at $\delta\in\mcA$ in the direction of $w\in\mcA$ is
\begin{align*}
&\langle \mathcal{D}\,J(\delta),w\rangle=
\mathbb{E}\left[\int_0^T  w_t\,\phi_t(\delta_t)\,\left(\delta_t-2\,\gamma\,D^{\delta}_{t^-}-\gamma-\alpha-2\,\gamma\,\mathbb{E}_{t^-}\left[\int_t^T\int_{\mathbb{R}}G(\delta_s-z')\,\rcompr(\diff z',\,\diff s)\right]\right)\,\diff A_t\right]\,,
\end{align*}
and vanishes  in every direction $w\in\mcA$   if and only if there is a process $\delta^*\in\mcA$ such that
\begin{align}\label{First_App_FBSDE}
\delta^*_t&=2\,\gamma\,\mathbb{E}_{t^-}\left[D^{\delta^*}_{T}\right]+\gamma+\alpha\,,
\end{align}
almost everywhere in $\mfT\times\Omega$.  
\end{theorem}
\begin{proof}
By Lemma \ref{Gat:CostLPandQP} and the performance criterion \eqref{eqn:performance criterion short}, the G\^ateaux derivative of $J$ vanishes at
\begin{align}
\delta^*_t&=2\,\gamma\,\mathbb{E}_{t^-}\left[\int_t^T\int_{\mathbb{R}}G(\delta^*_s-z')\,\rcompr(\diff z',\,\diff s)\right]+2\,\gamma\,\left(D^{\delta^*}_{t^-}+\tfrac{1}{2}\right)+\alpha\nonumber\\
&=2\,\gamma\,\mathbb{E}_{t^-}\left[\int_t^T\int_{\mathbb{R}}G(\delta^*_s-z')\,\rmes(\diff z',\,\diff s)\right]+2\,\gamma\,\left(D^{\delta^*}_{t^-}+\tfrac{1}{2}\right)+\alpha\nonumber\\
&=2\,\gamma\,\mathbb{E}_{t^-}\left[D^{\delta^*}_{T}-D^{\delta^*}_{t^-}\right]+2\,\gamma\,\left(D^{\delta^*}_{t^-}+\tfrac{1}{2}\right)+\alpha\nonumber\\
&=2\,\gamma\,\mathbb{E}_{t^-}\left[D^{\delta^*}_{T}\right]+\gamma+\alpha\label{FBSDE_App_raw}\,.
\end{align}
Now we  show that if the G\^ateaux derivative at $\delta$ vanishes in every direction $w$, the control $\delta$ satisfies \eqref{FBSDE_App_raw}. We proceed by contradiction.
Suppose there exists $\hat{\delta}\in\mcA$ such that $\langle \mathcal{D}\,J(\hat{\delta}),w\rangle=0$ for all $w\in\mcA$ and there is $(\mathbb{T},\,\mfO)\in \mathcal{B}(\mfT)\times\mcF_T$ with $\mathbb{L}(\mathbb{T})\,\mathbb{P}(\mfO)>0$ such that $\hat{\delta}_t(\omega)\neq \delta^*_t(\omega) $ for $t\in \mathbb{T}$, and $\omega\in\mfO$, where  $\mathbb{L}(\mathbb{T})$ denote the Lebesgue measure of $\mathbb{T}\in\mathcal{B}(\mfT)$, and $\mathcal{B}(\mfT)$ is the Borel sigma-algebra of $\mfT$. Thus, on  $\mathbb{T}\times \mfO$ we have
\begin{equation*}
\hat\delta_t(\omega)-2\,\gamma\,\mathbb{E}_{t^-}\left[D^{\hat\delta}_T\right](\omega)-\gamma-\alpha \neq 0\,.
\end{equation*}
Hence, $w_t=\hat\delta_t-2\,\gamma\,\mathbb{E}_{t^-} [D^{\hat\delta}_T ]-\gamma-\alpha$ is predictable  and $\mathbb{E} [\sup_{\tT}(w_t)^2 ]<\infty$. Furthermore, the G\^ateaux derivative of $\hat\delta$ in the direction of $w$  satisfies the inequality $\langle \mathcal{D}\,J(\hat{\delta}),w\rangle>0$, which is a contradiction. Therefore,  there is no $(\mathbb{T},\,\mfO)\in \mathcal{B}(\mfT)\times\mcF$ with $\mathbb{L}(\mathbb{T})\,\mathbb{P}(\mfO)>0$ such that $\hat{\delta}_t(\omega)\neq \delta^*_t(\omega)$ for $t\in \mathbb{T}$ and $\omega\in\mfO$. 
\end{proof}

If the value of the quadratic penalty parameter $\gamma$ is zero, the candidate optimal control in \eqref{First_App_FBSDE}   has the simple closed-form expression
\begin{align}\label{simple control gamma zero}
\delta^*_t&= \alpha\,,
\end{align}
which is independent of the number of missed trades. Thus, for $\gamma=0$ the agent sends all MLOs with discretion  $\alpha$ to walk the LOB.

\section{Existence and Uniqueness of the FBSDE}\label{sec:existence and uniqueness}
To the best of our knowledge, the  FBSDE in  \eqref{First_App_FBSDE} is a new class of random measure driven FBSDEs, and there are no uniqueness or existence results in the extant literature. Therefore, in this section we prove existence and uniqueness of the solution of the FBSDE. For FBSDEs in a semimartingale setting see \cite{Antonelli93}. For fully coupled FBSDEs in the Brownian motion case see \cite{Peng99}. For an account of Brownian motion and Poisson processes in FBSDEs, see \cite{Zhen1999}. \cite{Jianming2000}, \cite{Confortola2013}, \cite{confortola2016}, and \cite{bandini2016PhD} study the framework of BSDEs and MPPs. For the study of FBSDEs that arise from vanishing G\^ateaux derivatives in stochastic games stemming from algorithmic trading problems, see \cite{JaimungalPhilippe18} and \cite{JaimungalPhilippe19}.

To streamline the results in this section, we start with a lemma that is useful to prove existence and uniqueness of the solution to the FBSDE \eqref{First_App_FBSDE}.

\begin{lemma}
Let
\begin{align*}
\mcC:=\left\{ U = (U_t)_{\tT} \left| U\text{ is }\mathcal{F}-\text{{adapted} }\; \&\;
\mathbb{E}\left[\sup_{\tT}\left(U_t\right)^2 \right]\,<\,\infty \right. \right\}\,.
\end{align*}
The spaces $(\mcA,\norm{\cdot}_{\infty})$, $(\mcA,\norm{\cdot})$, $(\mcC,\norm{\cdot}_{\infty})$, and $(\mcC,\norm{\cdot})$ are Banach spaces, where \begin{equation*}
\norm{\delta}_{\infty}=\mathbb{E}\left[\sup_{\tT}\abs{\delta_t}\right]\qquad\text{and}\qquad \norm{\delta}=\mathbb{E}\left[\int_0^T\abs{\delta_t}\,\diff t\right]\,.
\end{equation*}
\end{lemma}
\begin{proof}
We prove the results  for the set $\mcA$ -- the proof for the set $\mcC$ is similar.

The predictable class of processes is closed in the space of finite processes with norm  $\norm{\cdot}_{\infty}$ (resp. $\norm{\cdot}$), which we denote by $L^{\infty}$ (resp. $L^{1}$). Then the space $(\mcA,\norm{\cdot}_{\infty})$ (resp. $(\mcA,\norm{\cdot})$) is a linear closed subspace of $L^{\infty}$ (resp. $L^{1}$), which is a Banach space and $(\mcA,\norm{\cdot}_{\infty})$ (resp. $(\mcA,\norm{\cdot})$) is also a Banach space. 
\end{proof}
  
\begin{corollary}\label{CxC_Banach}
The space $\mcC\times\mcC$ with norm 
\begin{equation*}
\norm{(U,V)}_{\mcC\times \mcC}=\norm{U}_{\mcC}+\norm{V}_{\mcC}\,,\quad\text{where}\quad \norm{U}_{\mcC}=\mathbb{E}\left[\int_0^T \abs{U_s}\diff s\right]
\end{equation*}
and $(U,V)\in\mcC\times \mcC$, is a Banach space.
\end{corollary}

By means of the change of variables $\tilde{\delta}_t=\delta^*_t+2\,\gamma\,D^{\delta^*}_t$, we have that a solution to the FBSDE
\begin{align}\label{FBSDE_s1}
{\delta}_t&=2\,\gamma\,\mathbb{E}_{t^-}\left[D^{{\delta}}_{T}\right]+\gamma+\alpha\,,\\
D^{{\delta}}_t&=\int_0^t \int_{\mathbb{R}}G({\delta}_{s}-z)\,\rmes(\diff z,\diff s)\,,\quad\quad\quad D^{\delta}_0=0\,,\nonumber
\end{align}
with $\delta\in\mcA$ and $D\in\mcC$, exists and is unique,  if and only if a solution  to the FBSDE
\begin{align}\label{FBSDE_s2}
\tilde{\delta}_t&=2\,\gamma\,\mathbb{E}_{t}\left[D^{\tilde{\delta}}_{T}-D^{\tilde{\delta}}_{t}\right]+\gamma+\alpha\,,\\
D^{\tilde{\delta}}_t&=\int_0^t \int_{\mathbb{R}}G(\tilde{\delta}_{s^-}+2\,\gamma\,D^{\tilde{\delta}}_{s^-}-z)\,\rmes(\diff z,\diff s)\,,\quad\quad\quad D^{\tilde{\delta}}_0=0\,,\nonumber
\end{align}
with $\tilde{\delta},\,D^{\tilde{\delta}}\in\mcC$, exists and is unique. We  write \eqref{FBSDE_s2} as
\begin{align}\label{eqn:FBSDE we study}
\tilde{\delta}_t&=2\,\gamma\,\mathbb{E}_{t}\left[\int_t^T \int_{\mathbb{R}}G(\tilde{\delta}_{s^-}+2\,\gamma\,D^{\tilde{\delta}}_{s^-}-z)\,\rmes(\diff z,\diff s)\right]+\gamma+\alpha\,,&\\
D^{\tilde{\delta}}_t&=\int_0^t \int_{\mathbb{R}}G(\tilde{\delta}_{s^-}+2\,\gamma\,D^{\tilde{\delta}}_{s^-}-z)\,\rmes(\diff z,\diff s)\,,&  D^{\tilde{\delta}}_0=0\,.\nonumber
\end{align}

To analyse solutions to the FBSDE \eqref{eqn:FBSDE we study}, we study the fixed points of the functional
\begin{align}\label{Functional_Full_FBSDE}
\Upsilon(U,V)_t&=\begin{pmatrix}
H(U,V)_t\\
I(U,V))_t
\end{pmatrix}
=\begin{pmatrix}
2\,\gamma\,\mathbb{E}_{t}\left[\int_t^T \int_{\mathbb{R}}G(U_{s^-}+2\,\gamma\,V^{}_{s^-}-z)\,\rmes(\diff z,\diff s)\right]+\gamma+\alpha\\
\int_0^t \int_{\mathbb{R}}G(U_{s^-}+2\,\gamma\,V^{}_{s^-}-z)\,\rmes(\diff z,\diff s)
\end{pmatrix},
\end{align}
and, for completeness,  prove existence and uniqueness of the solution of: (i)  the backward part of the FBSDE; (ii)  the forward part of the FBSDE; and (iii)  the full FBSDE -- a result which we derive independently from the existence of the backward and forward parts of the FBSDE.

The following theorem shows the existence and uniqueness of the solution to the backward part of the FBSDE \eqref{eqn:FBSDE we study}.

\begin{theorem}
Fix $V\in\mcC$. Let the cumulative distribution function $\Phi$ be Lipschitz with constant $k$, and let $\bar{\lambda}$ be the upper bound of the stochastic intensity $\lambda$ in Assumption \ref{Assumption_S.I.}. The functional $\Psi:\mcC\to\mcC$ given by
\begin{equation*}
\Psi(U)_t=2\,\gamma\,\mathbb{E}_{t}\left[\int_t^T \int_{\mathbb{R}}G(U_{s^-}+2\,\gamma\,V^{}_{s^-}-z)\,\rmes(\diff z,\diff s)\right]+\gamma+\alpha\,,\qquad V\in \mcC\,, 
\end{equation*}
has a unique fixed point.
\end{theorem}

\begin{proof}
We proceed as in Proposition A1 in \cite{duffie1992}. Define $Z=\sup_\tT\abs{X_t-Y_t}$ and $Z_t=\mathbb{E}_{t}\left[Z\right]$ for any $X$ and $Y$ in $\mcC$. Let $\Psi^{(1)}=\Psi$ and $\Psi^{(n)}=\Psi(\Psi^{(n-1)})$. Then
\begin{align*}
\abs{\Psi(X)_t-\Psi(Y)_t}&= 2\,\gamma \,\abs{\mathbb{E}_{t}\left[\int_t^T \left(\Phi(Y_{s^-}+2\,\gamma\,V_{s^-})-\Phi(X_{s^-}+2\,\gamma\,V_{s^-})\right) \,\lambda_s\,\diff s\right]}\\
&\leq 2\,\gamma \,k\,\mathbb{E}_{t}\left[\int_t^T \abs{X_{s^-}-Y_{s^-}} \,\lambda_s\,\diff s\right]\\
&\leq 2\,\gamma\,k \,\bar{\lambda}\,(T-t)\,\mathbb{E}_{t}\left[Z\right]\,.
\end{align*}
Use Fubini's theorem for conditional expectations to write
\begin{align*}
\abs{\Psi^{(2)}(X)_t-\Psi^{(2)}(Y)_t}&\leq 2\,\gamma\,k \,\bar{\lambda}\,\mathbb{E}_{t}\left[\int_t^T \abs{\Psi(X_s)-\Psi(Y_s)}\,\diff s\right]\\
&\leq 2\,\gamma\,k \,\bar{\lambda}\,\mathbb{E}_{t}\left[\int_t^T 2\,\gamma \,k\,\bar{\lambda}\,(T-s)\,\mathbb{E}_{s}\left[Z\right]\diff s\right]\\
&\leq \left(2\,\gamma\,k\,\bar{\lambda} \right)^2\,\mathbb{E}_{t}\left[\int_t^T\,(T-s)\,\mathbb{E}_{s}\left[Z\right]\,\diff s\right]\\
&\leq \left(2\,\gamma\,k\,\bar{\lambda} \right)^2\,\frac{(T-t)^2}{2!}\,\mathbb{E}_{t}\left[Z\right]\,,
\end{align*}
which after $n$ iterations becomes
\begin{align*}
\abs{\Psi^{(n)}(X)_t-\Psi^{(n)}(Y)_t}&\leq \left(2\,\gamma\,k\,\bar{\lambda} \right)^n\,\frac{(T-t)^n}{n!}\,\mathbb{E}_{t}\left[Z\right]\,.
\end{align*}
Finally, 
\begin{align*}
\mathbb{E}\left[\sup_\tT \abs{\Psi^{(n)}(X)_t-\Psi^{(n)}(Y)_t}\right]&\leq \frac{\left(2\,\gamma\,k\,\bar{\lambda} \,T\right)^n}{n!}\,\mathbb{E}\left[\sup_\tT\,\mathbb{E}_{t}\left[Z\right]\right]\\
&\leq 4\,\frac{\left(2\,\gamma\,k\,\bar{\lambda} \,T\right)^n}{n!}\,\mathbb{E}\left[\sup_\tT\,\abs{X_t-Y_t}\right]\,.
\end{align*}
Therefore,  for $n$ sufficiently large, the function $\Psi^{(n)}$ is a contraction mapping in the Banach space $\mcC$ equipped with the supremum norm $(\mcC,\norm{\cdot}_{\infty})$. Thus, there exists a unique\footnote{Unique in the sense of indistinguishability.} process $U\in \mcC$ such that $\Psi^{(n)}(U)=U$ and because $\Psi^{(n)}(\Psi(U))=\Psi\left(\Psi^{(n)}((U))\right)=\Psi(U)$ and by uniqueness of the fixed point, we have $\Psi(U)=U$, which proves the existence of the fixed point for $\Psi$. Uniqueness of this fixed point for $\Psi$ follows from uniqueness of the fixed point in $\Psi^{(n)}$, which concludes the proof. 
\end{proof}

The next theorem shows the existence and uniqueness of the solution to the forward part of the FBSDE \eqref{eqn:FBSDE we study}.

\begin{theorem}
Fix $U\in\mcC$. Let the distribution function $\Phi$ be Lipschitz with constant $k$, and let $\bar{\lambda}$ be the upper bound of the stochastic intensity $\lambda$ in Assumption \ref{Assumption_S.I.}. The functional $\Theta:\mcC\to\mcC$ given by
\begin{equation*}
\Theta(V)_t=\int_0^t \int_{\mathbb{R}}G(U_{s^-}+2\,\gamma\,V^{}_{s^-}-z)\,\rmes(\diff z,\diff s)\,,\qquad U\in\mcC,
\end{equation*}
has a unique fixed point.
\end{theorem}

\begin{proof}
First we prove that $\Theta$ is a functional from $\mcC$ to $\mcC$. Let $U,\,V\in \mcC$. By definition, the function $\Theta(V)$ is adapted and because  $G\leq1$ we have
\begin{align*}
\mathbb{E}\left[\sup_\tT\abs{\Theta(V)_t} \right]&\leq\mathbb{E}\left[\rmes\left([0,T],\mathbb{R}\right) \right]<\infty\,.
\end{align*}
Thus, $\Theta(V)\in\mcC$. Next, denote  $\Theta^{n}=\Theta(\Theta^{n-1})$ with $\Theta^{0}=\Theta(0)$ and define $h_n:[0,T] \to \mathbb{R}$ as
\begin{equation*}
h_n(t)=\mathbb{E}\left[\sup_{0\leq s\leq t}\abs{\Theta^{(n+1)}_s-\Theta^{(n)}_s}\right]\,.
\end{equation*}
We find an upper bound for $h_n(t)$ as follows:
\begin{align*}
h_n(t)&=\mathbb{E}\left[\sup_{0\leq s\leq t}\abs{\Theta^{(n+1)}_s-\Theta^{(n)}_s}\right]\\
&=\mathbb{E}\left[\sup_{0\leq s\leq t}\abs{\int_0^s \int_{\mathbb{R}}\left(G\left(U_{u^-}+2\,\gamma\,\Theta^{(n)}_{u^-}-z\right)-G\left(U_{u^-}+2\,\gamma\,\Theta^{(n-1)}_{u^-}-z\right)\right)\,\rmes(\diff z,\diff u)}\right]\\
&\leq\mathbb{E}\left[\int_0^t \int_{\mathbb{R}}\abs{G\left(U_{u^-}+2\,\gamma\,\Theta^{(n)}_{u^-}-z\right)-G\left(U_{u^-}+2\,\gamma\,\Theta^{(n-1)}_{u^-}-z\right)}\,\rcompr(\diff z,\diff u)\right]\\
&=\mathbb{E}\left[\int_0^t\abs{ \Phi\left(U_{u^-}+2\,\gamma\,\Theta^{(n)}_{u^-}\right)-\Phi\left(U_{u^-}+2\,\gamma\,\Theta^{(n-1)}_{u^-}\right)}\,\lambda_t\,\diff t\right]\\
&\leq 2\,\gamma\,k\,\bar{\lambda}\,\mathbb{E}\left[\int_0^t\abs{\Theta^{(n)}_{u^-}-\Theta^{(n-1)}_{u^-}}\,\diff t\right]\\
&\leq 2\,\gamma\,k\,\bar{\lambda}\,\int_0^t h_n(s)\, \diff s\,.
\end{align*}
The above inequality, together with the observation that $h_0(T)=M<\infty$, implies
\begin{equation*}
0\leq h_n(T)\leq \frac{M\,\left(2\,\gamma\,k\,\bar{\lambda}\right)^n\,T^n}{n!}
\,,
\end{equation*}
and use  Markov's inequality to obtain the bound:
\begin{equation*}
\mathbb{P}\left(\sup_\tT\abs{\Theta^{(n+1)}_t-\Theta^{(n)}_t}\geq 2^{-n}\right)\leq \frac{M\,\left(2\,\gamma\,k\,\bar{\lambda}\right)^n\,T^n\,2^{2\,n}}{n!}
\xrightarrow{n\to \infty} 0.
\end{equation*}
By  Borel-Cantelli arguments, there is $\mfO\subset \Omega$  such that for all $\omega \in \mfO$ the functions $t \to \Theta^{(n)}_t(w)$ form a Cauchy sequence in the supremum norm of $\mcC$ with probability one. Thus,  $\forall\, \omega\,\in\,\mfO$ there is a function  $\Theta^*_t(\omega)$ such that $\Theta^{(n)}_t(\omega)$ converges uniformly to $\Theta^*_t(\omega)$ in $\mfT$. Furthermore, there is an adapted modification of $\Theta^*$ in $\Omega$. 

Thus, the process $\Theta^*$  is a fixed point of the mapping defined by $\Theta$, and therefore satisfies the forward part of the FBSDE. 
\end{proof}

Finally, the next theorem shows the existence and uniqueness of the solution to the FBSDE   \eqref{eqn:FBSDE we study}.

\begin{theorem}\label{TheoremFBSDE}
Let the cumulative distribution function $\Phi$ be Lipchitz with parameter $k$ such that 
\begin{equation*}
k\,T\,\bar{\lambda}\,\left(\max \{1\,,2\,\gamma\}\right)^2<1\,,
\end{equation*}
where $\bar{\lambda}$ is the upper bound of the stochastic intensity $\lambda$ in Assumption \ref{Assumption_S.I.}. There exists a unique solution to the FBSDE
\begin{subequations}
\label{FBSDE_E!}
\begin{align}
\tilde{\delta}_t&=2\,\gamma\,\mathbb{E}_{t}\left[\int_t^T \int_{\mathbb{R}}G(\tilde{\delta}_{s^-}+2\,\gamma\,D^{\tilde{\delta}}_{s^-}-z)\,\rmes(\diff z,\diff s)\right]+\gamma+\alpha\,,\\
D^{\tilde{\delta}}_t&=\int_0^t \int_{\mathbb{R}}G(\tilde{\delta}_{s^-}+2\,\gamma\,D^{\tilde{\delta}}_{s^-}-z)\,\rmes(\diff z,\diff s)\,, & D^{\tilde{\delta}}_0=0\,.
\end{align}
\end{subequations}
\end{theorem}
\begin{proof} 
Consider the functional $\Upsilon:\mcC\times \mcC\to\mcC\times \mcC$ defined in \eqref{Functional_Full_FBSDE}. 
By Corollary \ref{CxC_Banach}, $\mcC\times \mcC$ is a Banach space when equipped with the norm
\begin{equation*}
\norm{\Upsilon(U,V)}_{\mcC\times \mcC}=\norm{H(U,V)}_{\mcC}+\norm{I(U,V)}_{\mcC}\,,\quad\text{where}\quad \norm{U}_{\mcC}=\mathbb{E}\left[\int_0^T \abs{U_s}\diff s\right]\,.
\end{equation*}

Let $(U,V)$ and  $(X,Y)$ be in  $\mcC\times \mcC$ and write
\begin{align}
\label{eqn:norm upsilon}
\norm{\Upsilon(U,V)-\Upsilon(X,Y)}_{\mcC\times \mcC}=& 
\;
\mathbb{E}\left[\int_0^T\abs{H(U,V)_t-H(X,Y)_t}\diff t\right]
\nonumber
\\
&
+ \mathbb{E}\left[\int_0^T\abs{I(U,V)_t-I(X,Y)_t}\diff t\right]\,.
\end{align}
The first term on the right-hand side of \eqref{eqn:norm upsilon} satisfies the bound  
\begin{align*}
&\mathbb{E}\left[\int_0^T\abs{H(U,V)_t-H(X,Y)_t}\diff t\right]\\
&\quad\leq \mathbb{E}\left[\int_0^T 2\,\gamma\,\mathbb{E}_{t}\left[\int_t^T \abs{\Phi(U_{s^-}+2\,\gamma\,V^{}_{s^-})-\Phi(X_{s^-}+2\,\gamma\,Y^{}_{s^-})}\,\diff A_s\right]\, \diff t\right]\\
&\quad\leq 2\,k\,\gamma\,\bar{\lambda}\,\int_0^T\mathbb{E}\left[ \int_t^T \abs{U_{s^-}+2\,\gamma\,V^{}_{s^-}-X_{s^-}-2\,\gamma\,Y^{}_{s^-}}\diff s\right]\, \diff t\,.
\end{align*}
The second term on the right-hand side of \eqref{eqn:norm upsilon} satisfies the bound
\begin{align*}
&\mathbb{E}\left[\int_0^T\abs{I(U,V)_t-I(X,Y)_t}\diff t\right]\\
&\quad\leq\int_0^T \mathbb{E}\left[\int_0^t \int_{\mathbb{R}}\abs{\,G(U_{s^-}+2\,\gamma\,V^{}_{s^-}-z)-G(X_{s^-}+2\,\gamma\,Y^{}_{s^-}-z)\,}\,\rcompr(\diff z,\diff s)\right]\, \diff t\\
&\quad=\int_0^T \mathbb{E}\left[\int_0^t \abs{\,\Phi(U_{s^-}+2\,\gamma\,V^{}_{s^-})-\Phi(X_{s^-}+2\,\gamma\,Y^{}_{s^-})\,}\,\lambda_s\,\diff s\right] \diff t\\
&\quad\leq k\,\bar{\lambda}\,\int_0^T\mathbb{E}\left[ \int_0^t \abs{\,U_{s^-}+2\,\gamma\,V^{}_{s^-}-X_{s^-}-2\,\gamma\,Y^{}_{s^-}\,}\diff s\right]\, \diff t\,.
\end{align*}
Now, let $k_1=k\,\bar{\lambda}\,\max \{2\,\gamma,\, 1\}$ and $k_2=k_1\,\max \{2\,\gamma,\,1\}$, and write
\begin{align*}
\norm{\Upsilon(U,V)-\Upsilon(X,Y)}_{\mcC\times \mcC}&\leq k_1\,\int_0^T\mathbb{E}\left[ \int_0^T \abs{U_{s^-}+2\,\gamma\,V^{}_{s^-}-X_{s^-}-2\,\gamma\,Y^{}_{s^-}}\diff s\right] \diff t\\
&\leq k_1\,T\,\mathbb{E}\left[ \int_0^T \abs{U_{t^-}+2\,\gamma\,V_{t^-}-X_{t^-}-2\,\gamma\,Y_{t^-}}\diff t\right]\\
&\leq k_2\,T\,\,\mathbb{E}\left[ \int_0^T \abs{U_{t^-}-X_{t^-}}\diff t\right]+k_2\,T\,\mathbb{E}\left[ \int_0^T \abs{V_{t^-}-Y_{t^-}}\diff t\right]\\
&<\norm{(U,V)-(X,Y)}_{\mcC\times \mcC}\,.
\end{align*}
Thus, $\Upsilon$ is a contraction mapping in the Banach space $\mcC\times \mcC$ (see Corollary \ref{CxC_Banach}), so there exists a unique pair of processes $U^*$ and $V^*$ such that $\Upsilon(U^*,V^*)=(U^*,V^*)$.
\end{proof}

In all, we have shown that the candidate optimal control in \eqref{First_App_FBSDE} exists and is unique. Finally, it is straightforward to see that $\delta^*\in\mcA$. By definition, the control $\delta^*$ is predictable. A short calculation shows
\begin{align*}
\mathbb{E}\left[\sup_{0\leq t \leq T}\left(\delta^*_t\right)^2\right]&\leq \mathbb{E}\left[\sup_{0\leq t \leq T}\left(\mathbb{E}_t\left[N_T\right]\right)^2\right]\\
&=\mathbb{E}\left[\sup_{0\leq t \leq T}\left(N_t +\int_t^T\lambda_s\diff s\right)^2\right]
\leq 2\,\mathbb{E}\left[N^2_T\right] +2\,T\bar{\lambda}^2
<\infty\,.
\end{align*}
Therefore the control $\delta^*$ satisfying \eqref{First_App_FBSDE} is an element of $\mcA$.

\section{Optimality}\label{sec:optimality}
In this section we prove that the discretion $\delta^*$ satisfying \eqref{First_App_FBSDE}  is the global minimizer of the agent's performance criterion $J(\delta)$. We prove this in several steps. First, Theorem \ref{thm:delta is local minimum} shows that the control $\delta^*$ is a local minimum of $J(\delta)$. Then, after proving two auxiliary lemmas, Theorem \ref{thm:global optimality} shows that $\delta^*$ is the global minimizer of the performance criterion.

\begin{theorem}\label{thm:delta is local minimum} The control $\delta^*$ satisfying \eqref{First_App_FBSDE}  is a local minimum of the agent's performance criterion $J(\delta)$.
\end{theorem}
\begin{proof}
Recall that the G\^ateaux derivative  $\langle \mathcal{D}\,J(\delta^*),w\rangle$ vanishes in every direction $w\in\mcA$.  The second G\^ateaux derivative\footnote{See Appendix \ref{SecondGD} for details of the second G\^ateaux derivative.} at $\delta\in\mcA$  in the directions $\nu,w\in\mcA$  is 
\begin{subequations}
\begin{eqnarray}\label{eqn:second derivative}
\langle \mathcal{D}^2\,J(\delta),\,w,\,\nu\rangle &=&\mathbb{E}\left[\int_0^T w_t\,\nu_t\, \phi_t'(\delta_t)\left(\delta_t-2\,\gamma\,\mathbb{E}_{t^-}\left[D^{\delta}_T\right]-\gamma-\alpha\right)\diff A_t\right]\\ \label{eqn:second derivative b}
&& +\mathbb{E}\left[\int_0^T w_t\,\phi_t(\delta_t)\left(\nu_t+2\,\gamma\,\mathbb{E}_{t^-}\left[\int_0^T \phi_s(\delta_s)\,\nu_s\,\diff A_s\right]\right)\diff A_t\right]\,.
\end{eqnarray}
\end{subequations}
 This G\^ateaux derivative is non-negative at $\delta=\delta^*$ because the expression on the right-hand side of \eqref{eqn:second derivative}  is zero at $\delta=\delta^*$ and the expression in \eqref{eqn:second derivative b} is non-negative for every $\nu,\,w\,\in \mcA$. Therefore, $\delta^*$ is a local minimum.
\end{proof}

\begin{lemma}\label{LemmaContinuity}
Let $\delta\in\mcA$  and  $J(\delta)=y_0<\infty$. Let $\bar{\lambda}<\infty$ be the bound  for the stochastic intensity $\lambda$ in Assumption \ref{Assumption_S.I.}, and let  $\bar{N}<\infty$ be a bound for the number of trade attempts. Assume the function
\begin{equation*}
r_t(x)=\int_{-\infty}^x z\,\phi_t(\diff z)
\end{equation*}
is  Lipschitz in $x$ uniformly on $[0,T]\times \Omega$, with Lipschitz constant $\tilde{k}$.

Given $\epsilon>0$, define
\begin{equation}\label{eqn:eta}
\eta=\frac{\epsilon}{2}\left(
(\tilde{k}+4\,\bar{N}\,k)\,\bar{\lambda}\,T+2\,\bar{\lambda}^2\,k\,T^2
\right)^{-1}>0\,,
\end{equation}
then, for all $w\in\mcA$ such that $\norm{\delta-w}_{\infty}<\eta$, we have  $\abs{y_0-J(w)}<\epsilon$.
\end{lemma}
\begin{proof}
Consider $w\in\mcA$ s.t. $\norm{\delta-w}_{\infty}<\eta$. Recall, $J=J^C+\alpha\,J^{LP}+\gamma\,J^{QP}$ and observe that 
\begin{equation*}
\abs{J(\delta)-J(w)}\leq \abs{J^C(\delta)-J^C(w)}+\alpha\,\abs{J^{LP}(\delta)-J^{LP}(w)}+\gamma\,\abs{J^{QP}(\delta)-J^{QP}(w)}\,.
\end{equation*}
Next, we bound each term on the right-hand side of the inequality. Firstly,
\begin{align*}
\abs{J^C(\delta)-J^C(w)}&=\abs{\mathbb{E}\left[\int_0^T\int_{\mathbb{R}}z\,\left(\Ffct(\delta_t-z)-\Ffct(w_t-z)\right)\,\phi_t(\diff z)\,\lambda_t\,\diff t\right]}
\\
&
\leq \tilde{k}\,\mathbb{E}\left[\int_0^T \abs{\delta_t-w_t} \diff A_t \right]
\\
&
\leq \tilde{k}\,\bar{\lambda}\,T\,\mathbb{E}\left[\sup_\tT\abs{\delta_t-w_t}\right].
\end{align*}
Secondly,
\begin{equation*}
\abs{J^{LP}(\delta)-J^{LP}(w)}
\leq \mathbb{E}\left[\int_0^T\abs{\Phi(\delta_t)-\Phi(w_t)}\, \diff A_t\right]
\leq \bar{\lambda}\,k\,T\,\mathbb{E}\left[\sup_\tT\abs{\delta_t-w_t}\right].
\end{equation*}
Finally,
\begin{align}
\abs{J^{QP}(\delta)-J^{QP}(w)} &\leq 
\mathfrak{A}^{\delta,w}
+ k\,\mathbb{E}\left[\int_0^T\abs{\delta_t-w_t} \diff A_t\right]
\label{eqn: ineq third term}
\leq
\mathfrak{A}^{\delta,w}+ k\,\bar{\lambda}\,T\,\mathbb{E}\left[\sup_\tT\abs{\delta_t-w_t}\right],
\end{align}
where
\[
\mathfrak{A}^{\delta,w} :=\abs{\mathbb{E}\left[\int_0^T\int_{\mathbb{R}}\left(2\,G(\delta_t-z)\,D^{\delta}_{t^-}-2\,G(w_t-z)\,D^{w}_{t^-}\right)\phi_t(\diff z)\, \diff A_t\right]}.
\]
Next, we bound  the first term on the right-hand side of inequality \eqref{eqn: ineq third term}:
\begin{align*}
\mathfrak{A}^{\delta,w}
&\leq \abs{\mathbb{E}\left[\int_0^T\int_{\mathbb{R}}2\,G(\delta_t-z)\left(D^{\delta}_{t^-}-D^{w}_{t^-}\right)\,\phi_t(\diff z) \,\diff A_t\right]}
\\\nonumber &\quad\quad\quad+\abs{\mathbb{E}\left[\int_0^T\int_{\mathbb{R}}2\,D^{w}_{t^-}\left(G(\delta_t-z)-G(w_t-z)\right)\,\phi_t(\diff z) \,\diff A_t\right]}\\
&\quad\quad\leq 2\,\bar{\lambda}\,\mathbb{E}\left[\int_0^T\abs{D^{\delta}_{t^-}-D^{w}_{t^-}}\diff t\right]+2\,\bar{N}\,k\,\bar{\lambda}\,\mathbb{E}\left[\int_0^T \abs{\delta_t-w_t}\diff t\right]\\
&\quad\quad\leq 2\,\bar{\lambda}\,\mathbb{E}\left[\int_0^T\abs{\int_0^{t^-}\int_{\mathbb{R}}G(\delta_s-z)-G(w_s-z)\,\rmes(\diff z,\diff s) }\diff t\right]+2\,\bar{N}\,k\,T\,\bar{\lambda}\,\mathbb{E}\left[\sup_\tT \abs{\delta_t-w_t}\right]\\
&\quad\quad\leq 2\,\bar{\lambda}\,k\,\int_0^T\mathbb{E}\left[\int_0^{t^-}\abs{\delta_s-w_s}\,\lambda_t\diff s\right] \diff t+2\,\bar{N}\,k\,T\,\bar{\lambda}\,\mathbb{E}\left[\sup_\tT \abs{\delta_t-w_t}\right]\\
&\quad\quad\leq 2\,\bar{\lambda}^2\,k\,T^2\,\mathbb{E}\left[\sup_\tT \abs{\delta_t-w_t}\right]+2\,\bar{N}\,k\,T\,\bar{\lambda}\,\mathbb{E}\left[\sup_\tT \abs{\delta_t-w_t}\right]\\
&\quad\quad= (2\,\bar{\lambda}^2\,k\,T^2+2\,\bar{N}\,k\,T\,\bar{\lambda})\,\mathbb{E}\left[\sup_\tT \abs{\delta_t-w_t}\right]\,.
\end{align*}
Hence,
\begin{align*}
\abs{J^{QP}(\delta)-J^{QP}(w)}&\leq \left(2\,\bar{\lambda}^2\,k\,T^2+2\,\bar{N}\,k\,T\,\bar{\lambda} +k\,T\,\bar{\lambda}\right)\,\mathbb{E}\left[\sup_\tT \abs{\delta_t-w_t}\right]\,,
\end{align*}
and since $\norm{\delta-w}_{\infty}\le\eta$,  we have
\begin{align*}
\abs{J(\delta)-J(w)}&\leq \abs{J^C(\delta)-J^C(w)}+\alpha\,\abs{J^{LP}(\delta)-J^{LP}(w)}+\gamma\,\abs{J^{QP}(\delta)-J^{QP}(w)}\\
&\leq \left(  \tilde{k}\,\bar{\lambda}\,T+ \alpha\,k\,\bar{\lambda}\,T+2\,\gamma\,\bar{\lambda}^2\,k\,T^2+k\,\gamma\,T\,\bar{\lambda}+2\,\gamma\,\bar{N}\,k\,T\,\bar{\lambda}\right)\,\mathbb{E}\left[\sup_\tT \abs{\delta_t-w_t}\right]\\
&\leq \left(\bar{z}\,k\,T\,\bar{\lambda}+2\,\bar{\lambda}^2\,k\,T^2+4\,k\,T\,\bar{\lambda}\right)\,\eta\\
&=\frac{\epsilon}{2}<\epsilon\,,
\end{align*}
where the last equality follows from the choice of $\eta$ in \eqref{eqn:eta}, and the proof is complete.
\end{proof}

Before proving the main result of this section, which shows that our candidate control is the global minimum of the performance criterion $J(\delta)$, we prove the following auxiliary lemma. 
\begin{lemma}\label{LemmaGlobalMinimum}
If the functional $J(\delta)$ has a global minimum $\hat{\delta}\in \mcA$, then
\begin{equation}
\langle \mathcal{D}\,J(\hat{\delta}),w\rangle\geq 0\,,\quad\quad \forall w \in \mcA\,.
\end{equation}
\end{lemma}
\begin{proof}
The proof is by contradiction. Suppose there is $\hat{w}\in\mcA$ such that $\langle \mathcal{D}\,J(\hat{\delta}),\hat{w}\rangle=\hat{\eta}<0$. Set $\epsilon=\abs{\nicefrac{\hat{\eta}}{2}}>0$, and because
\begin{equation}
\hat{\eta}= \lim_{\epsilon\to 0}\frac{J(\hat{\delta}+\epsilon \,\hat{w})-J(\hat{\delta})}{\epsilon}\,,
\end{equation}
there exists $\rho>0$ such that if $\abs{\epsilon}<\rho$, then
\begin{equation}
\abs{\frac{J(\hat{\delta}+\epsilon \,\hat{w})-J(\hat{\delta})}{\epsilon}-\hat{\eta}}<\abs{\hat{\eta}/2}\,.
\end{equation}
Now, fix $\hat{\epsilon}$ such that $0<\hat{\epsilon}<\rho$, then 
\begin{equation}
\frac{J(\hat{\delta}+\hat{\epsilon} 
\,\hat{w})-J(\hat{\delta})}{\hat{\epsilon}}<\abs{ \hat{\eta}/2}+\hat{\eta}= \hat{\eta}/2<0\,.
\end{equation}
Therefore,
\begin{equation}\label{OptimalGlobalEq}
J(\hat{\delta}+\hat{\epsilon} \,\hat{w}) <\hat{\epsilon}\,\frac{\hat{\eta}}{2}+J(\hat{\delta})<J(\hat{\delta})\,,
\end{equation}
and because $\hat{\delta}, \hat{w}\in \mcA$, the control  $\hat{\delta}+\hat{\epsilon} \,\hat{w}$ is in the set $\mcA$, and by \eqref{OptimalGlobalEq}, we have the inequality $J(\hat{\delta}+\hat{\epsilon} \,\hat{w})<J(\hat{\delta})$, which contradicts $\hat{\delta}$ being a global minimizer.
\end{proof}

\begin{theorem}\textbf{Global optimality.}\label{thm:global optimality}
If $J$ has a global minimum at $\hat{\delta}\in \mcA$, then $\hat{\delta}=\delta^*$ a.e. in $\mfT\times \Omega$, with $\delta^*$ solving \eqref{First_App_FBSDE}.
\end{theorem}
\begin{proof}
The proof is by contradiction.  Suppose the global minimum $\hat{\delta}\in \mcA$, but it is not true that $\hat{\delta}=\delta^*$ a.e. in $\mfT\times \Omega$, with $\delta^*$ solving \eqref{First_App_FBSDE}, i.e.,   there exists $(\mathbb{T},\,\mfO)\in \mathcal{B}(\mfT)\times\mcF_T$ with $\mathbb{L}(\mathbb{T})\,\mathbb{P}(\mfO)>0$ such that $\hat{\delta}\ne\delta^*$ on $\mathbb{T}\times\mfO$. First, by Lemma \ref{LemmaGlobalMinimum} 
\begin{equation}
\langle \mathcal{D}\,J(\hat{\delta}),w\rangle\geq 0\,,\qquad  \forall w \in \mcA\,,
\end{equation}
and because $\hat{\delta}\neq\delta^*$ on $\mathbb{T}\times\mfO$,  there exists $\hat{w}\in\mcA$ such that $\langle \mathcal{D}\,J(\hat{\delta}),\hat{w}\rangle> 0$.
Now, take $\tilde{w}=-\hat{w} \in \mcA$, then
\begin{align*}
\langle \mathcal{D}\,J(\hat{\delta}),\tilde{w}\rangle&= \mathbb{E}\left[\int_0^T  \tilde{w}_t\,\phi_t(\hat{\delta}_t)\left(\hat{\delta}_t-2\,\gamma\,\mathbb{E}_{t^-}\left[D^{\hat{\delta}}_{T}\right]-\alpha-\gamma \right)\,\diff A_t\right]\\
&= -\mathbb{E}\left[\int_0^T  \hat{w}_t\,\phi_t(\hat{\delta}_t)\left(\hat{\delta}_t-2\,\gamma\,\mathbb{E}_{t^-}\left[D^{\hat{\delta}}_{T}\right]-\alpha-\gamma \right)\,\diff A_t\right]\\
&=-\langle \mathcal{D}\,J(\hat{\delta}),\hat{w}\rangle\\
&<0\,,
\end{align*}
which contradicts Lemma \ref{LemmaGlobalMinimum}. Therefore, if there is a global minimum at $\hat{\delta}\in \mcA$, then $\hat{\delta}=\delta^*$ a.e. in $\mfT\times \Omega$.
\end{proof}

\section{Performance of strategy}\label{sec:numerical results} 
The expectation that appears in  \eqref{First_App_FBSDE} is conditional on the information $\mathcal{F}_{t^-}$, therefore the process $\delta^*$ is  a sub-martingale. Here, we study a slight variation of the FBSDE in \eqref{First_App_FBSDE} and derive a partial-integro differential equation for the optimal control. 

To this end, fix the optimal control $\delta^*\in \mcA$ and define the  process $(\check{\delta}_t)_{\tT}$, where
\begin{align*}
\check{\delta}_t=2\,\gamma\,\mathbb{E}_{t}\left[D^{\delta^*}_{T}\right]+\gamma+\alpha\,.
\end{align*}
Observe that $\delta^*$ in \eqref{First_App_FBSDE} is the c\`agl\`ad (LCRL) version of the c\`adl\`ag (RCLL) process $\check{\delta}$, and $\delta^*_t=\check{\delta}_{t^-}$. Define the dynamics of the missed trades $D^{\delta^*}$ as a function of the process $\check{\delta}$:
\begin{align*}
D^{{\check{\delta}}}_t
&=\int_0^t \int_{\mathbb{R}}G(\check{\delta}_{s^-}-z)\,\rcompr(\diff z,\diff s)+\int_0^t \int_{\mathbb{R}}G(\check{\delta}_{s^-}-z)\,\rcompd(\diff z,\diff s)\,,
\end{align*}
and recall that $\rcompd=\rmes-\rcompr$ is the compensated random measure of $\mathcal N$. 

\begin{Assumption}\label{Assumption Markov Intensity}
The stochastic intensity $\left(\lambda_t\right)_{t\in\mfT}$ has the Markov property, furthermore, the quadratic co-variation between the process $\lambda$ and $D^{\check{\delta}}$ is zero.
\end{Assumption}

By Assumption \ref{Assumption Markov Intensity},  we derive the Markov property of $\check{\delta}$, which we use to  write  $\check{\delta}_t=h(t,D^{\check{\delta}}_t,\lambda_t)$ for a differentiable function $h$ with respect to the first argument. Then the process $D^{\check{\delta}}$ is given by
\begin{equation*}
D^{\check{\delta}}_t = \int_0^t \int_{\mathbb{R}}G(h(s,D^{\check{\delta}}_{s^-},\lambda_t)-z)\,\rcompr(\diff z,\diff s)+\int_0^t \int_{\mathbb{R}}G(h(s,D^{\check{\delta}}_{s^-},\lambda_t)-z)\,\rcompd(\diff z,\diff s)\,,
\end{equation*}
and because $\check{\delta}$ is a martingale, the function $h$ is the solution of a PIDE that we characterize in the following theorem.

\begin{theorem}\label{PIDE Theorem}
Let $\check{\delta}_t=h(t,D^{\check{\delta}}_t,\lambda_t)$. Under Assumptions \ref{Assumption_S.I.} and \ref{Assumption Markov Intensity}, the function $h$  satisfies the PIDE
\begin{align}\label{eqn:PIDE Theorem}
0=\partial_t h(t,D,\lambda)+\mathcal{L}^{\lambda}_t h(t,D,\lambda)+\left(\int_{h(t,D,\lambda)}^{\infty} \lambda\, \phi_t(z) \, \diff z\right)\left(h(t,D+1,\lambda)-h(t,D,\lambda)\right)\,,
\end{align}
with boundary and terminal conditions
\begin{align*}
&\lim_{D \to \infty} h(t,D,\lambda)=\infty \qquad\text{ and }\qquad h(T,D,\lambda)=2\,\gamma\,D+\gamma+\alpha\,.
\end{align*}
Here, $\mathcal{L}^{\lambda}_t h(t,D,\lambda)$ is the infinitesimal generator of the arrival intensity process $\lambda$ acting on the function $h$.
\end{theorem}
\begin{proof}

Apply It\^o's formula to $\check{\delta}_t=h(t,D_t^{\check\delta},\lambda_t)$ and note that the drift term (i.e.,  the $\diff t$-term)  vanishes because $\check{\delta}$ is a martingale. Existence and uniqueness of a solution to this PIDE follow from a comparison principle. Specifically, we have
\begin{align*}
0&=\partial_t h(t,D,\lambda)+\mathcal{L}^{\lambda}_t h(t,D,\lambda)+\left(\int_{h(t,D,\lambda)}^{\infty} \lambda_t\, \phi_t(z)\,\diff z\right)\,\left(h(t,D+1,\lambda)-h(t,D,\lambda)\right)\\
&\leq \partial_t h(t,D,\lambda)+\mathcal{L}^{\lambda}_t h(t,D,\lambda)+\left(\int_{-\infty}^{\infty} \lambda_t\, \phi_t(z)\,\diff z\right)\,\left(h(t,D+1,\lambda)-h(t,D,\lambda)\right)\\
&= \partial_t h(t,D,\lambda)+\mathcal{L}^{\lambda}_t h(t,D,\lambda)+ \lambda_t \,\left(h(t,D+1,\lambda)-h(t,D,\lambda)\right)\\
&\leq \partial_t h(t,D,\lambda)+\mathcal{L}^{\lambda}_t h(t,D,\lambda)+ \bar{\lambda} \,\left(h(t,D+1,\lambda)-h(t,D,\lambda)\right)\,. 
\end{align*}
\end{proof}

We use the continuity of $h$ in $t$ to write $\delta^*_t=h(t,D_{t^-}^{\check\delta},\lambda_{t^-})$, see   characterization for $h$ in \eqref{eqn:PIDE Theorem} to compute $\delta^*_t=\check{\delta}_{t^-}$.

\subsection{Poisson arrival of trades}
We solve the PIDE in \eqref{eqn:PIDE Theorem} numerically  to illustrate the performance of the latency-optimal strategy. Assume the agent sends MLOs  according to a homogeneous Poisson process with intensity $\lambda =100$,  the linear penalty parameter is $\alpha=0$, the quadratic penalty parameter  $\gamma$ takes values in $\{0.01,\,0.03,\,0.1\}$, the marks (price and quantity shocks to the LOB) are iid normal  $Z_n\sim \mathbf N(0.2,1)$, $n=1,2,\dots$, and the trading horizon is  $T=1$.

Figure \ref{Fig:DeltaStar h function}  shows the discretion $\delta^*$  as a function of the number of missed trades. The left panel shows three surfaces, one for each value of the quadratic penalty parameter $\gamma$. The higher the value of the quadratic penalty parameter for missing trades, the higher is the optimal discretion employed in the strategy.  The right panel shows the optimal discretion when the number of missed trades is $D^{{\delta^*}}\in\{4,\,8,\,12\}$, and the quadratic penalty parameter is $\gamma\in\{0.01,0.03,0.1\}$. Blue denotes cases with $\gamma=0.01$, green for $\gamma=0.03$, and  red for $\gamma=0.1$. Solid lines are for $D^{\delta^*}=4$, dashed lines are for $D^{\delta^*}=8$, and dash-dotted lines are for $D^{\delta^*}=12$.

\begin{figure}
\begin{center}
\includegraphics[width=0.45\textwidth]{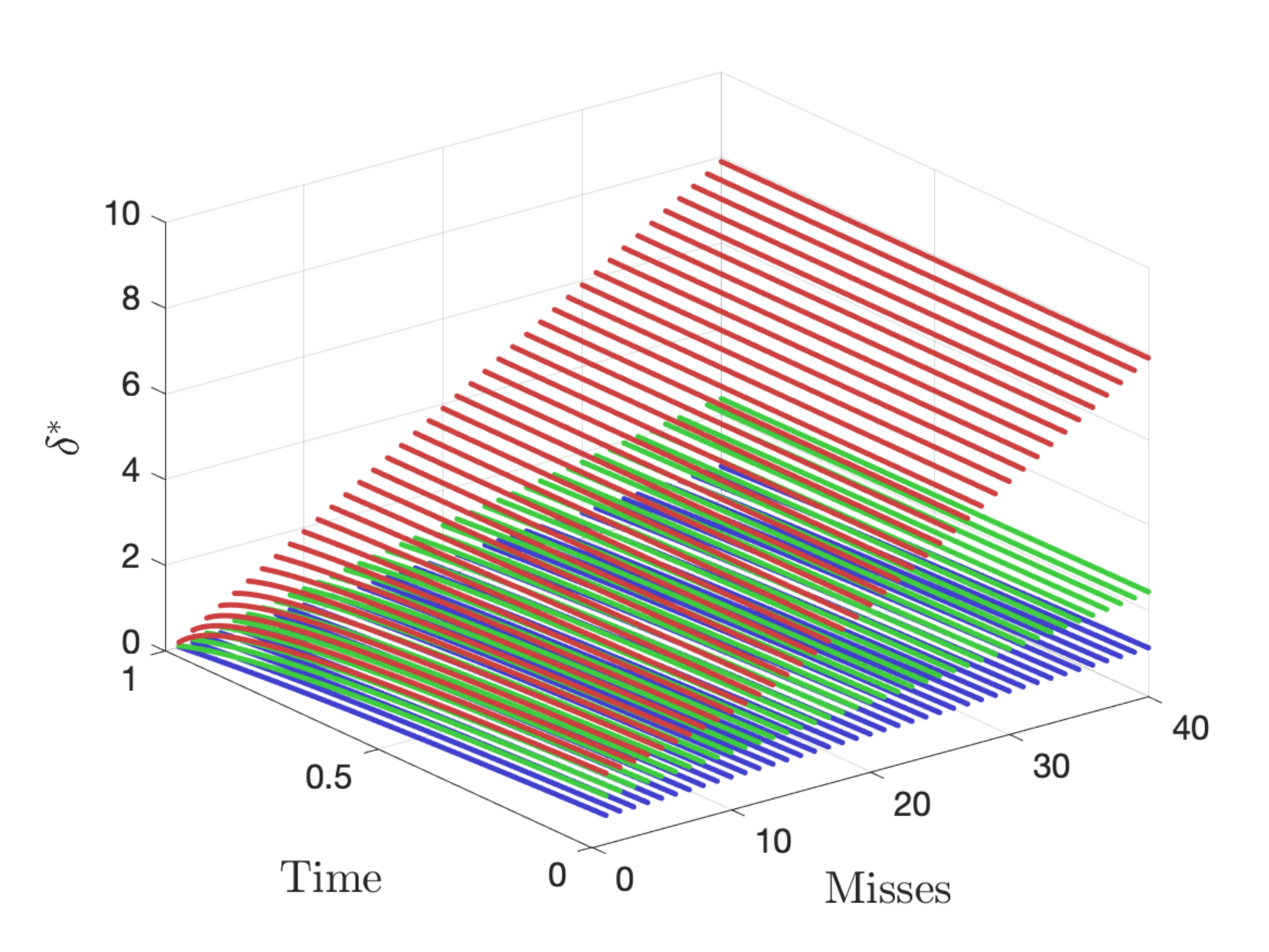}
\includegraphics[width=0.44\textwidth]{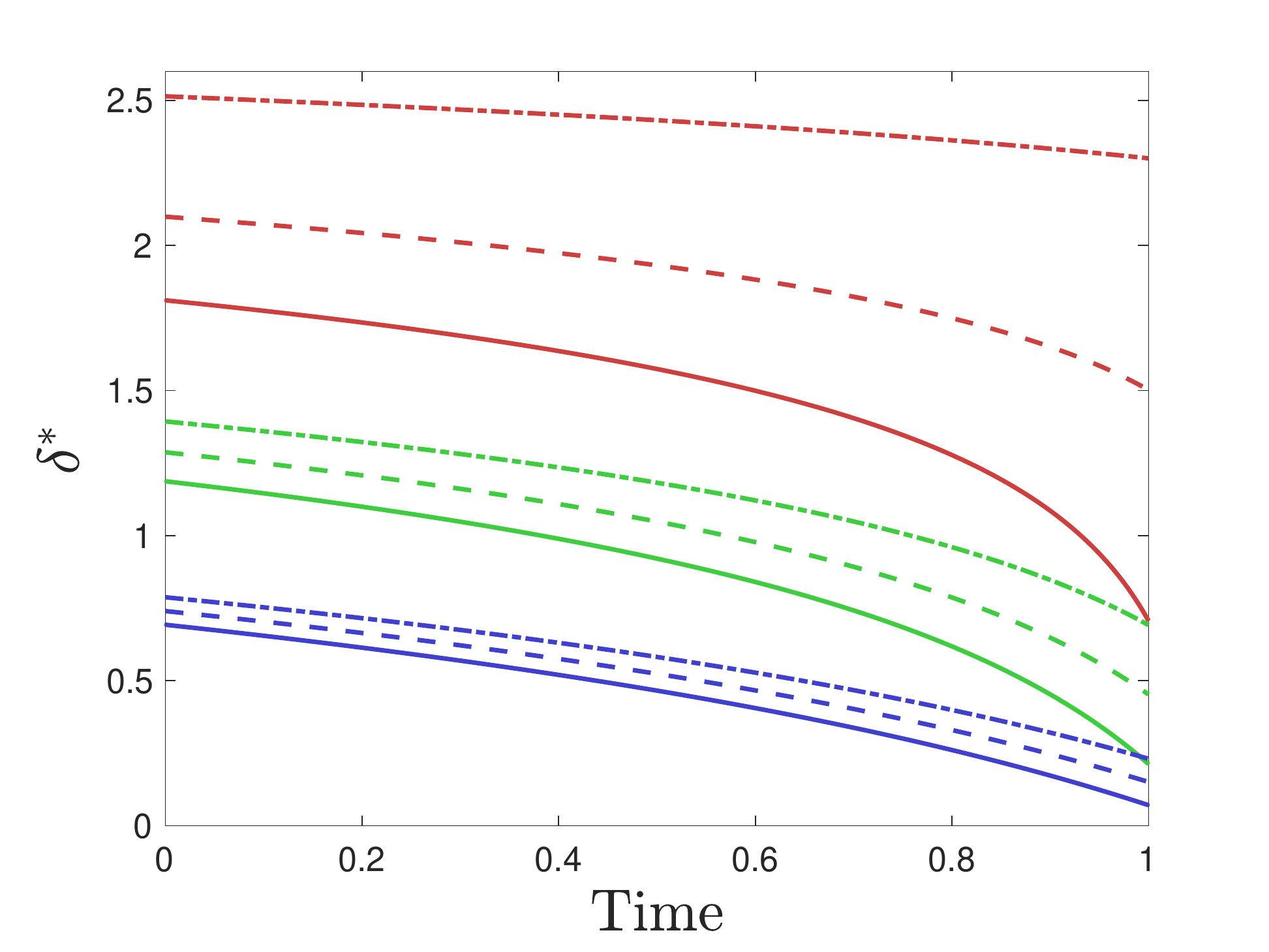}
\caption{Left panel: Optimal strategy $\delta^*$ as a function of time and the number of missed trades  for $\gamma=0.01$ (bottom surface), $\gamma=0.03$ (middle surface), and $\gamma=0.1$ (top surface). The remaining parameters are: $\lambda=100$, $\alpha=0$, and $Z_n\sim \mathbf N(0.2,1)$ for every $n$. Right panel: Optimal strategy for various values of missed trades; blue curves are for $D=4$, green curves are for $D=8$, and red curves are for $D=12$.}
\label{Fig:DeltaStar h function}
\end{center}
\end{figure}
We perform 10,000 simulations of the agent's trading activity and Figure \ref{Fig:Paths_Simulation} shows three sample paths. The top panel shows the optimal discretion of the agent's orders and the cumulative costs accrued from walking the book and from receiving price improvements. The bottom panel shows the number of missed trades and the number of trade attempts. Clearly, as the number of missed trades increases (decreases), the optimal strategy is to increase (decrease) the discretion of the MLOs to walk the LOB.
\begin{figure}[H]
\begin{center}
\includegraphics[width=0.80\textwidth]{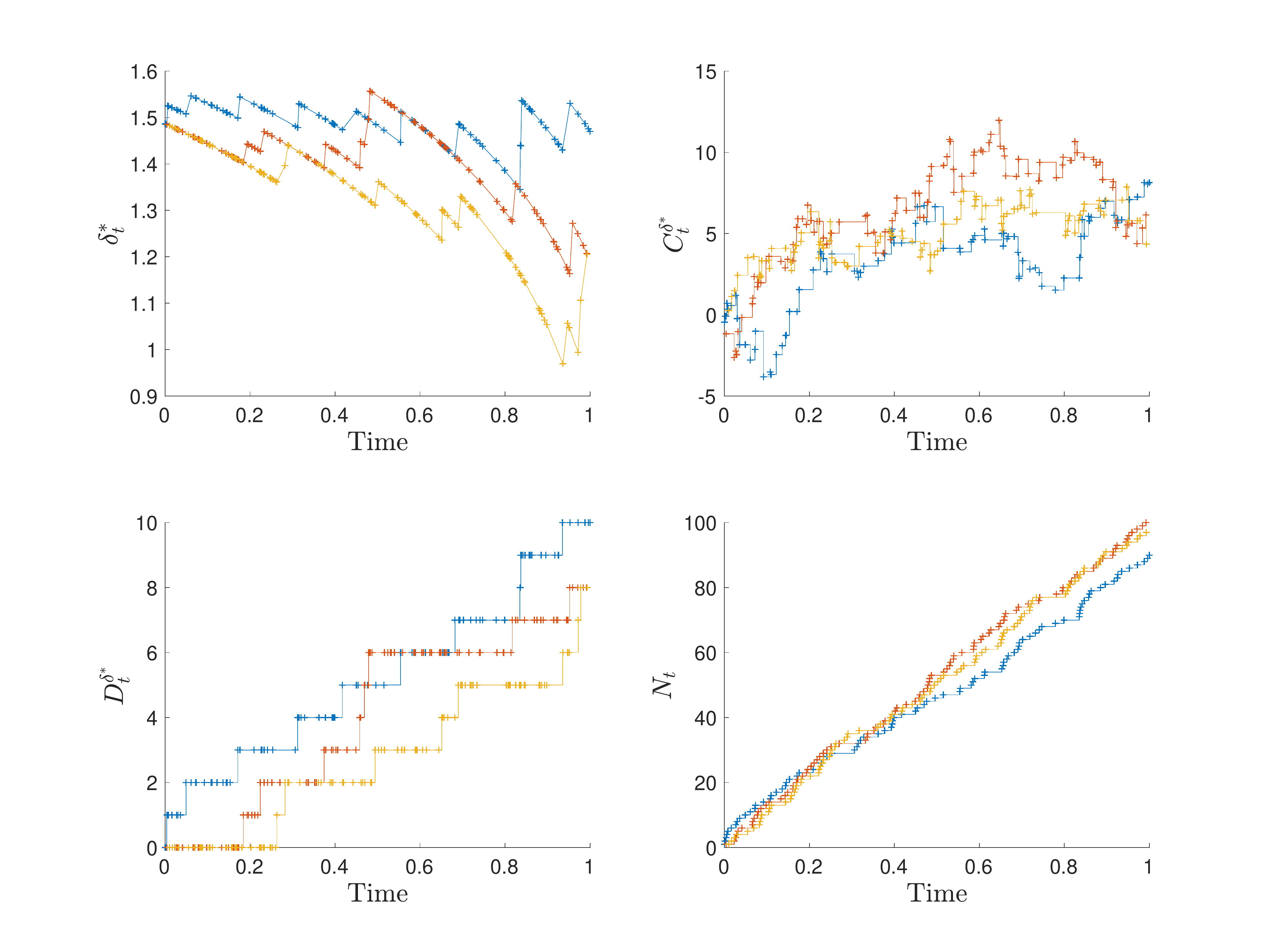}
\caption{Sample paths for the optimal discretion $\delta^*$ (top left panel), number of missed trades $D^{\delta^*}$ (lower left panel), cost of strategy $C^{\delta^*}$ (top right panel), and number of trade attempts $N$ (lower right panel) for three simulations of the MPP. Parameters: $\alpha=0$, $\gamma = 0.07$, $\lambda =100$, $T=1$. }
\label{Fig:Paths_Simulation}
\end{center}
\end{figure}

Figure \ref{Fig:Histograms_DiffGammas} reports various cost metrics of the optimal strategy  for three values of the quadratic penalty parameter $\gamma$. The top panel shows histograms of the  cost incurred by the strategy to fill trades, i.e., $C^{\delta^*}_T$,  and the average cost of walking the LOB to fill trades, i.e.,  $C^{\delta^*}_T/(N_T-D^{\delta^*}_T)$.  Recall that the cost is negative (positive) when the trade is executed with price improvement (deterioration). The Figure shows that as the value of the quadratic penalty parameter increases: (i) the average cost of walking the book to fill trades increases, the total cost increases, and the average number of misses decreases, see bottom panels; (ii) the costs of walking the LOB increase because the strategy fills more orders (i.e., misses fewer trades), see the bottom-left panel. The bottom-right panel shows that the average ratio of missed trades to trade attempts decreases when the penalty for missing trades increases.

\begin{figure}[H]
\begin{center}
\includegraphics[width=0.36\textwidth]{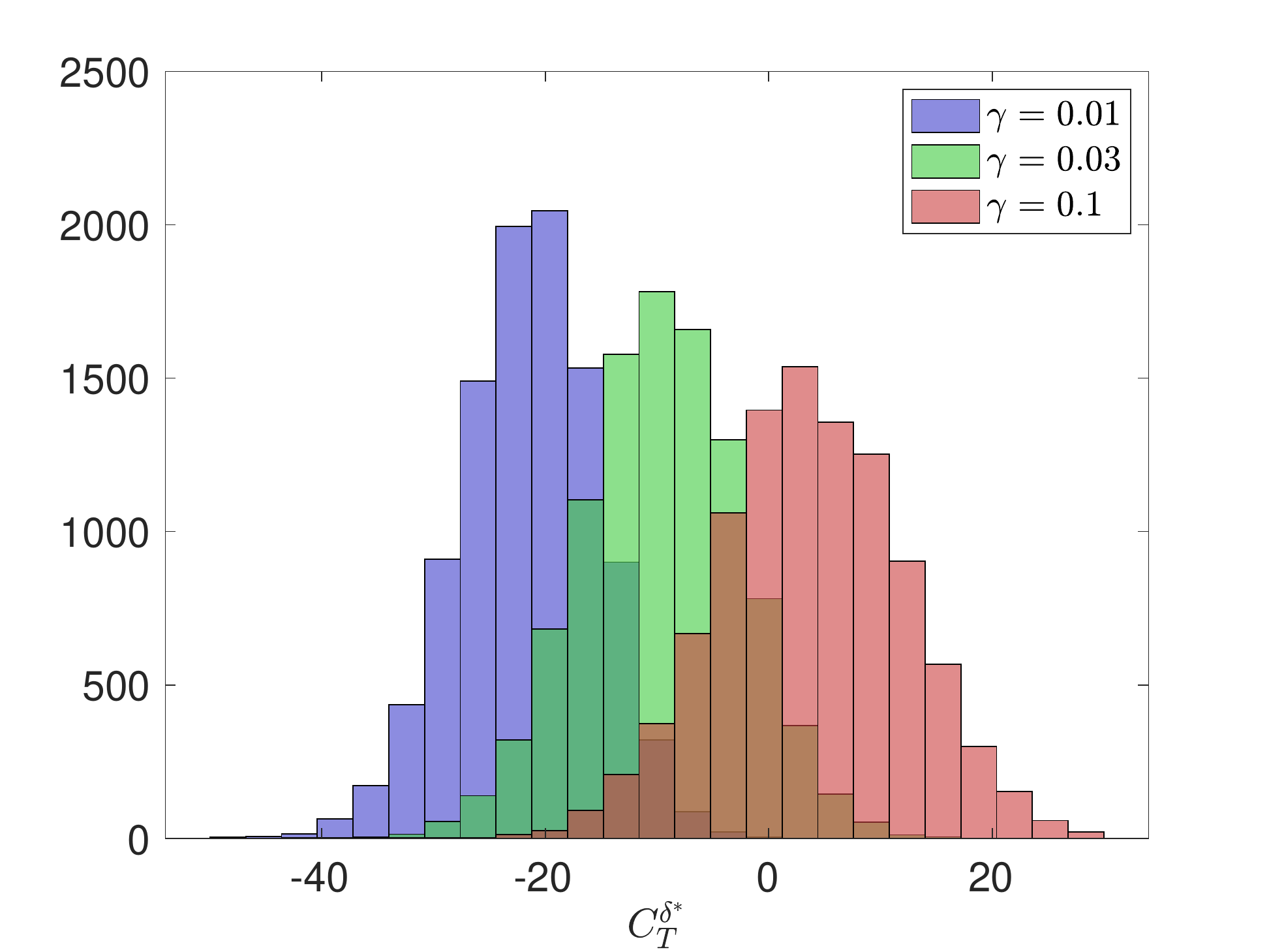}
\includegraphics[width=0.36\textwidth]{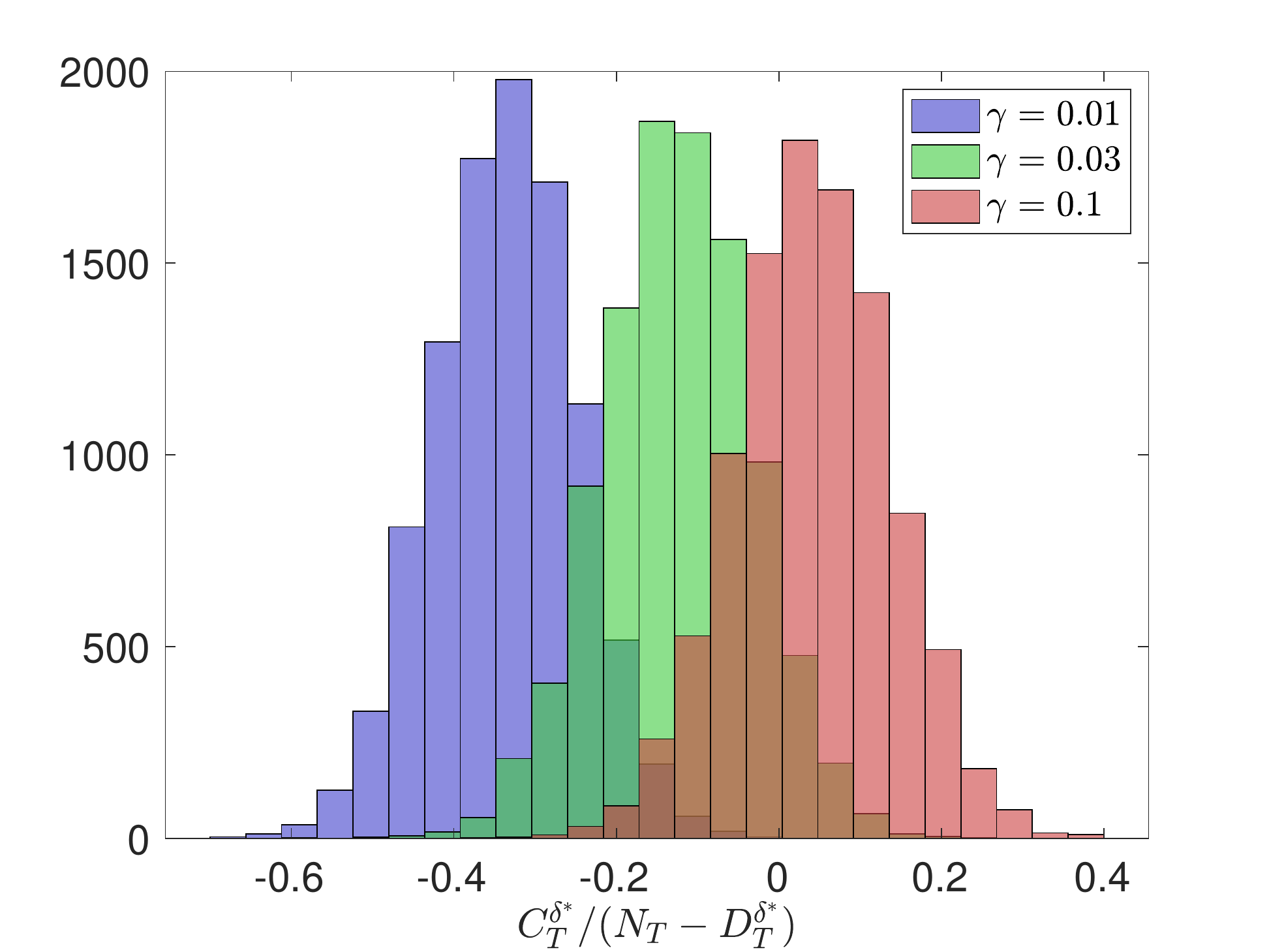}
\includegraphics[width=0.36\textwidth]{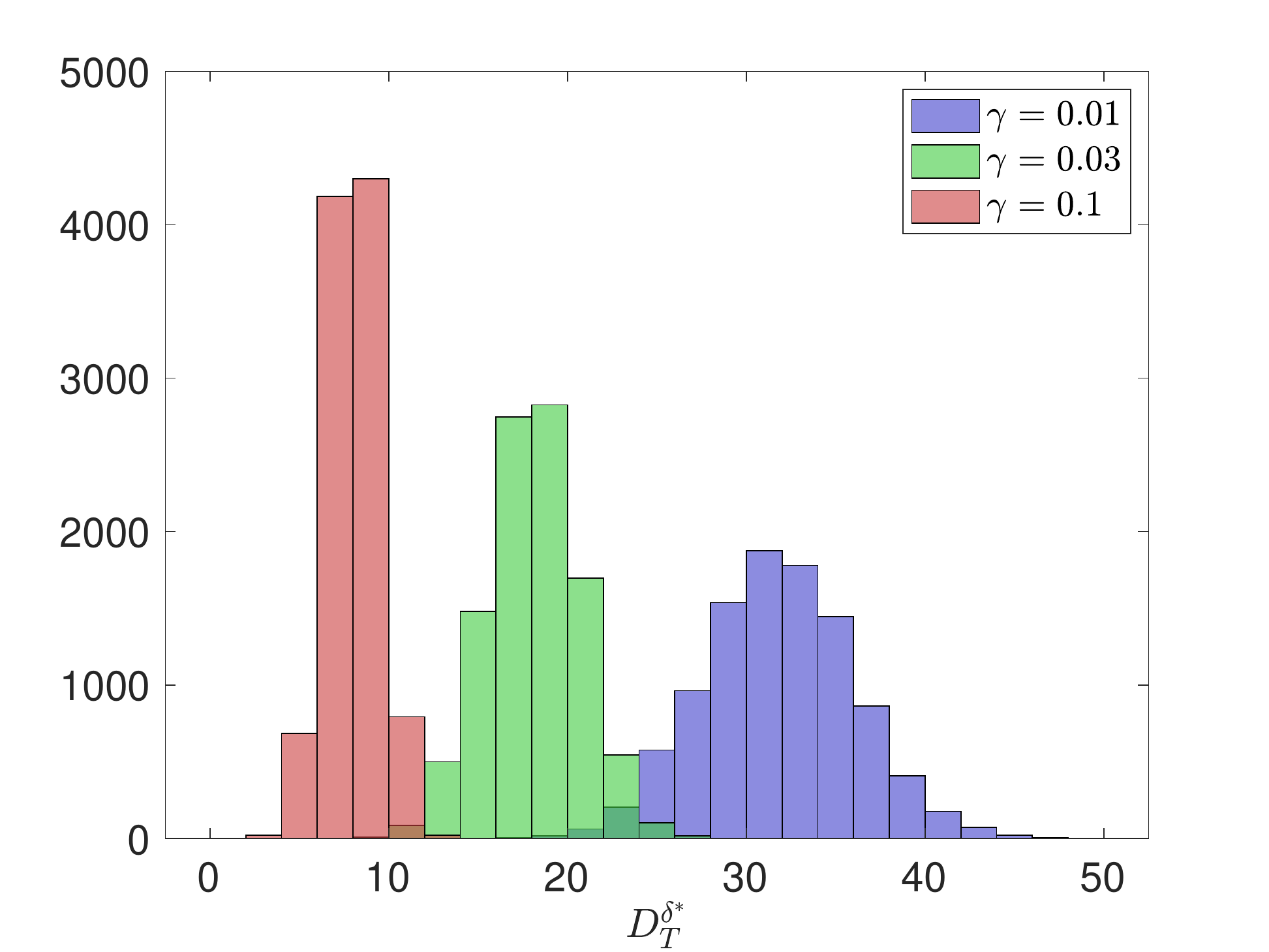}
\includegraphics[width=0.36\textwidth]{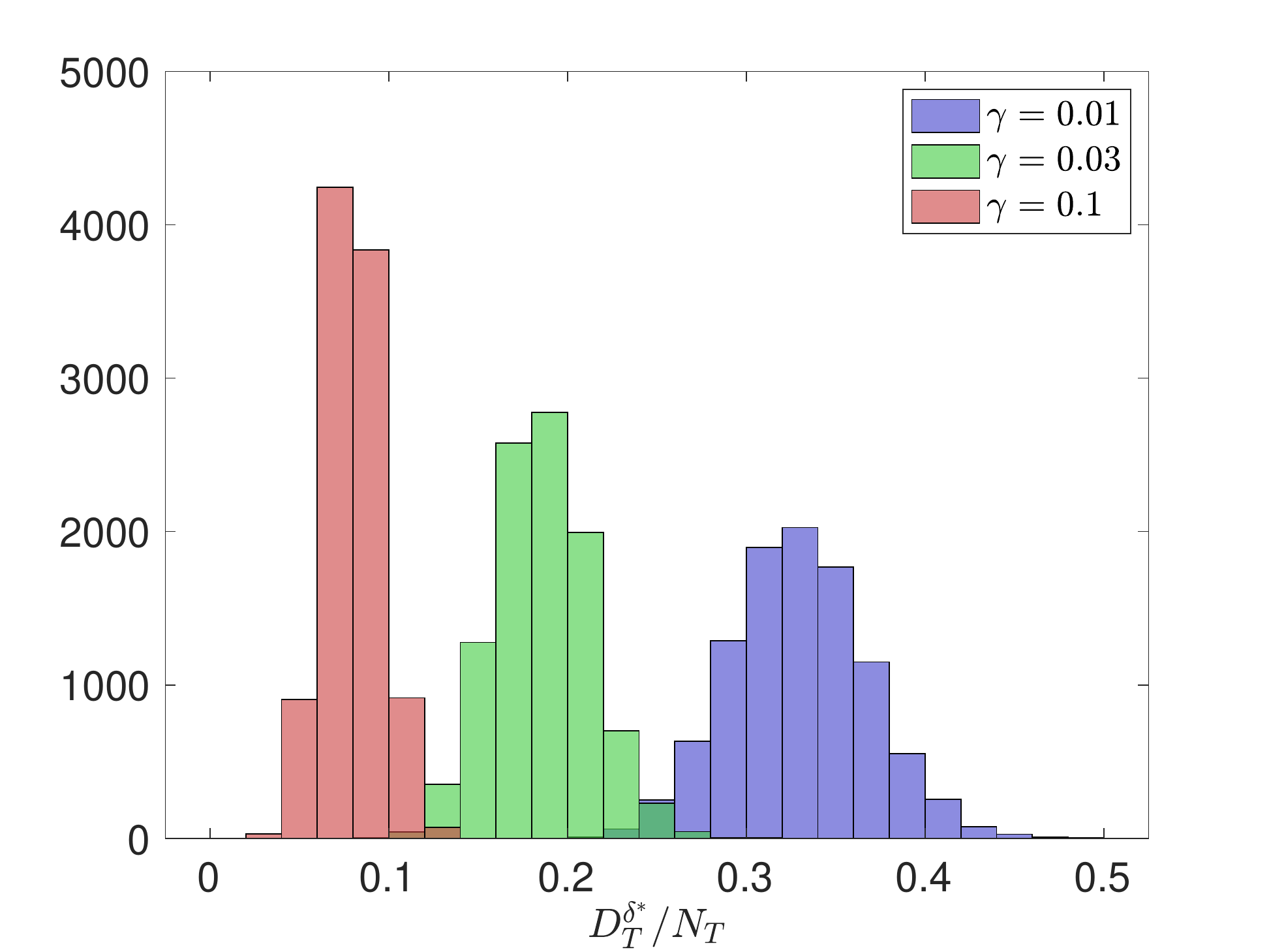}

\caption{Top left panel: Histogram of the cost $C^{\delta^*}_T$ of the strategy. Top right panel: Histogram of the extra cost per filled trade $C^{\delta^*}_T /  (N_T - D^{\delta^*}_T)$. Bottom left panel: Histogram of the number of misses $D^{\delta^*}_T$. Bottom right panel: Histogram of percentage of misses $D^{\delta^*}_T/N_T$. }
\label{Fig:Histograms_DiffGammas}
\end{center}
\end{figure}

The tradeoff between higher fill ratios and costs of walking the book are clear. An agent who seeks very high fill ratios, i.e., high values of   $(N_T-D_T^{\delta^*})/N_T$,  employs very high values of the penalty parameters in the performance criterion.  Other agents  may prefer to swap price improvements for price deteriorations in their overall trading strategy. For example, in the 10,000 simulations we discuss,  when $\gamma\approx 0.0693$ the average cost of filled trades,  $ C^{\delta^*}_T /  (N_T - D^{\delta^*}_T) $,   is zero and the average rate of missed trades,  $ D^{\delta^*}_T/N_T$ is  0.1048.

Finally, a naive strategy employed by liquidity takers is to send MLOs with no discretion to walk the LOB, see \cite{CarteaLeo2018}.  Here, the expected ratio of missed trades to number of attempts and the expected cost of the strategy for an agent who sends all MLOs with no discretion to walk the LOB is $\mathbb{E}[D^{0}_T/N_T]=0.5797 $ and  $\mathbb{E}[C^{0}_T]=-29.25$, respectively. The expected cost is negative because the strategy does not accrue costs from walking the book, but may receive price improvements.

\subsubsection{Optimal vs fixed discretion to walk the LOB}
We compare the results of a strategy with $\alpha=0$ and $\gamma>0$ with those of a fixed discretion latency-optimal strategy (i.e., $\alpha >0$ and $\gamma =0$). Recall that when $\gamma=0$ the optimal strategy is independent of the number of misses, so the agent sends  all MLOs with discretion  $\delta^* = \alpha$, see \eqref{simple control gamma zero}.

The top panels in Figure \ref{Fig:Risk Metric} show  the probability that the number of missed trades is less than 10\% of trade attempts, i.e.,  $\mathbb P [D^{\delta^*}_T < 0.1\, N_T]$,   and the expected cost of the strategy, i.e.,   $\mathbb E[C^{\delta^*} _T]$, when the agent sends orders with a fixed discretion to walk the LOB, i.e., $\gamma = 0$ and $\alpha\in[0,\,2.5]$. Similarly, the bottom panels  show   the probability that the number of missed trades is less than 10\% of trade attempts, i.e.,  $\mathbb P [D^{\delta^*}_T < 0.1\,N_T]$   and the expected cost of the strategy, i.e.,  $\mathbb E[C^{\delta^*} _T]$ for $\gamma\in[0.02,\,0.16]$ and $\alpha=0$. The orange circle in each picture shows the lowest expected terminal cost $\mathbb E[C^{\delta^*} _T]$ for which  $\mathbb P [D^{\delta^*}_T < 0.1\, N_T]\geq 0.95$.  The expected terminal cost of the  fixed discretion latency-optimal strategy with $\alpha= 1.91$ is approximately 9.52, and  the expected cost obtained with the latency-optimal strategy, with $\gamma= 0.127$, is approximately 5.93. 

Also,  the expected number of misses when $\gamma=0$ and $\alpha=1.9125$ (orange circle point in the top panels) is $\mathbb{E}\left[D^{\delta^*}_T\right]=6.39$, and when $\alpha=0$ and $\gamma=0.1220$ (orange circle point in the bottom panels) we obtain $\mathbb{E}\left[D^{\delta^*}_T\right]=5.41$.

Thus, an agent who does not expect to miss more than 10\% of the trades with high probability  may prefer a latency-optimal optimal strategy with $\gamma>0$ and $\alpha=0$ than a strategy that sends MLOs with a fixed discretion during the entire trading window.

\begin{figure}[H]
\begin{center}
\includegraphics[width=0.35\textwidth]{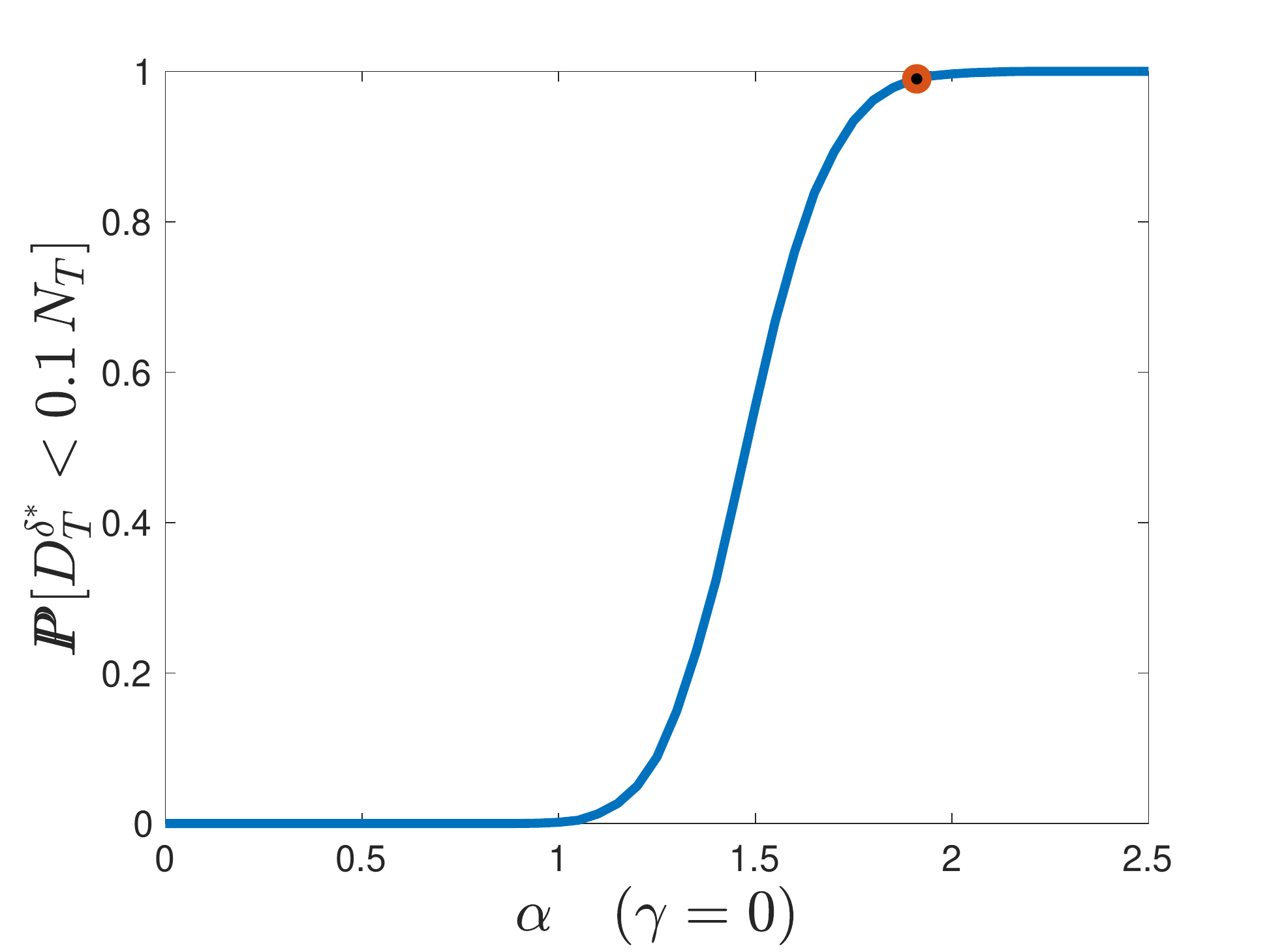}
\includegraphics[width=0.35\textwidth]{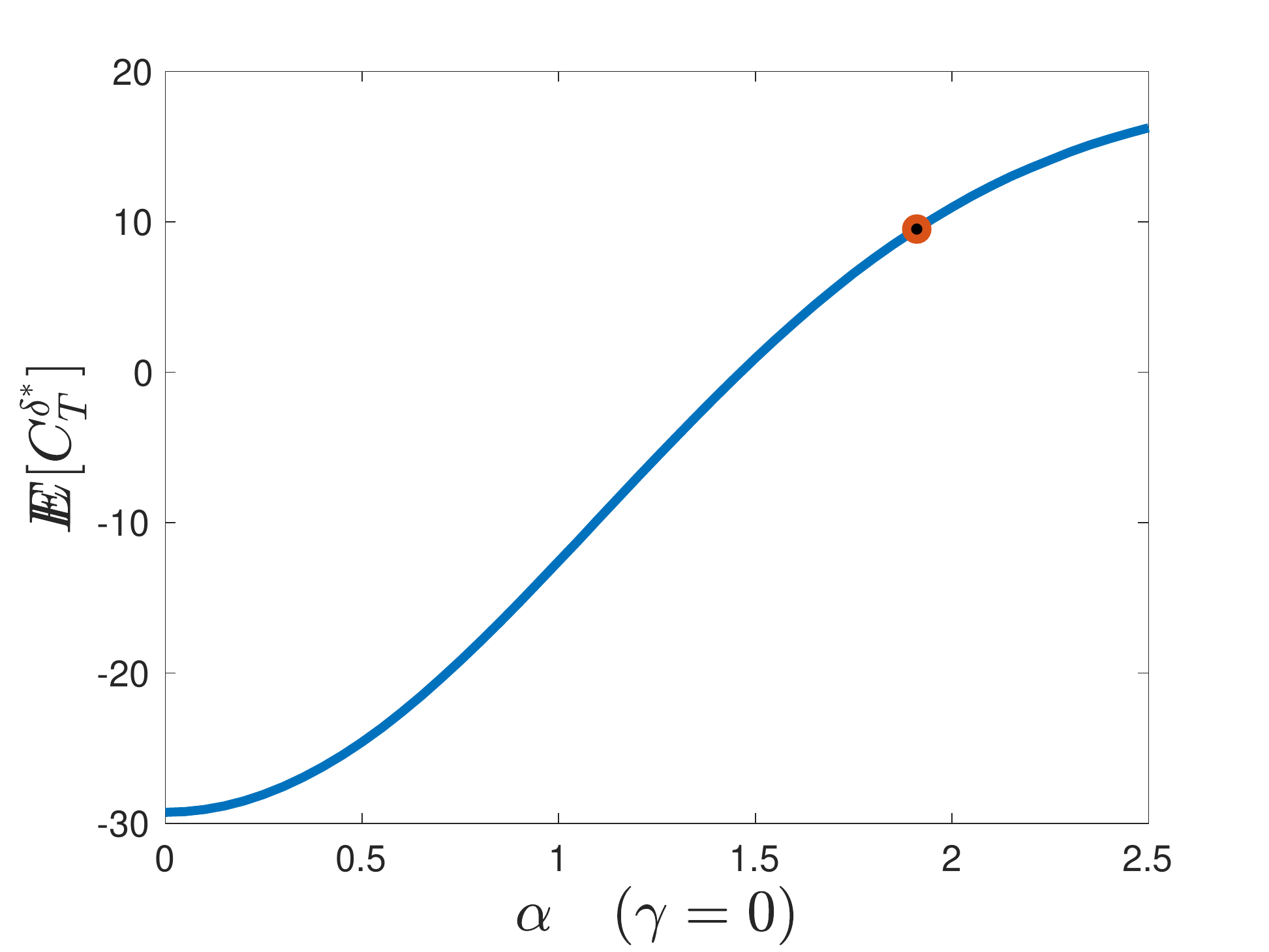}
\includegraphics[width=0.35\textwidth]{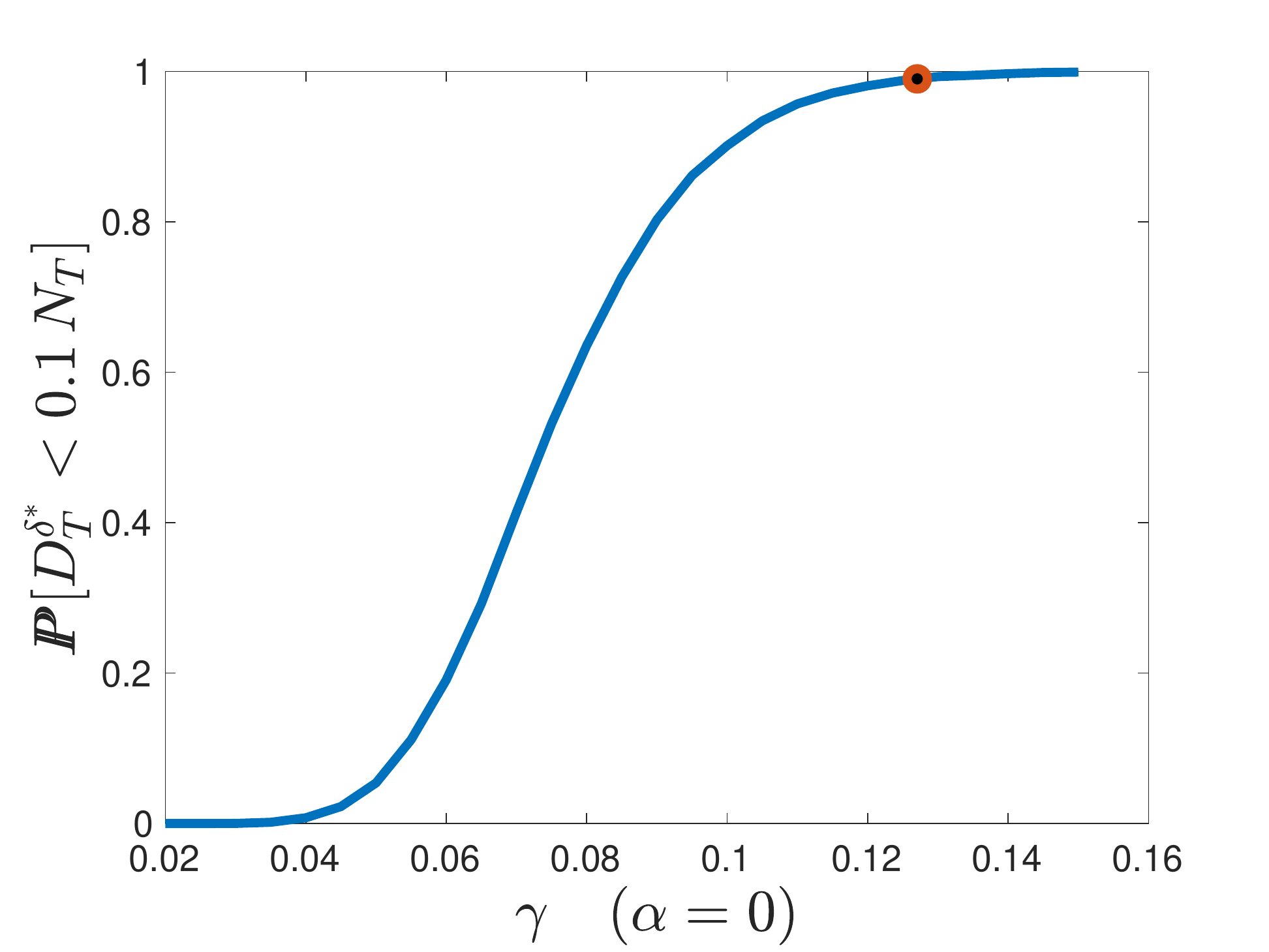}
\includegraphics[width=0.35\textwidth]{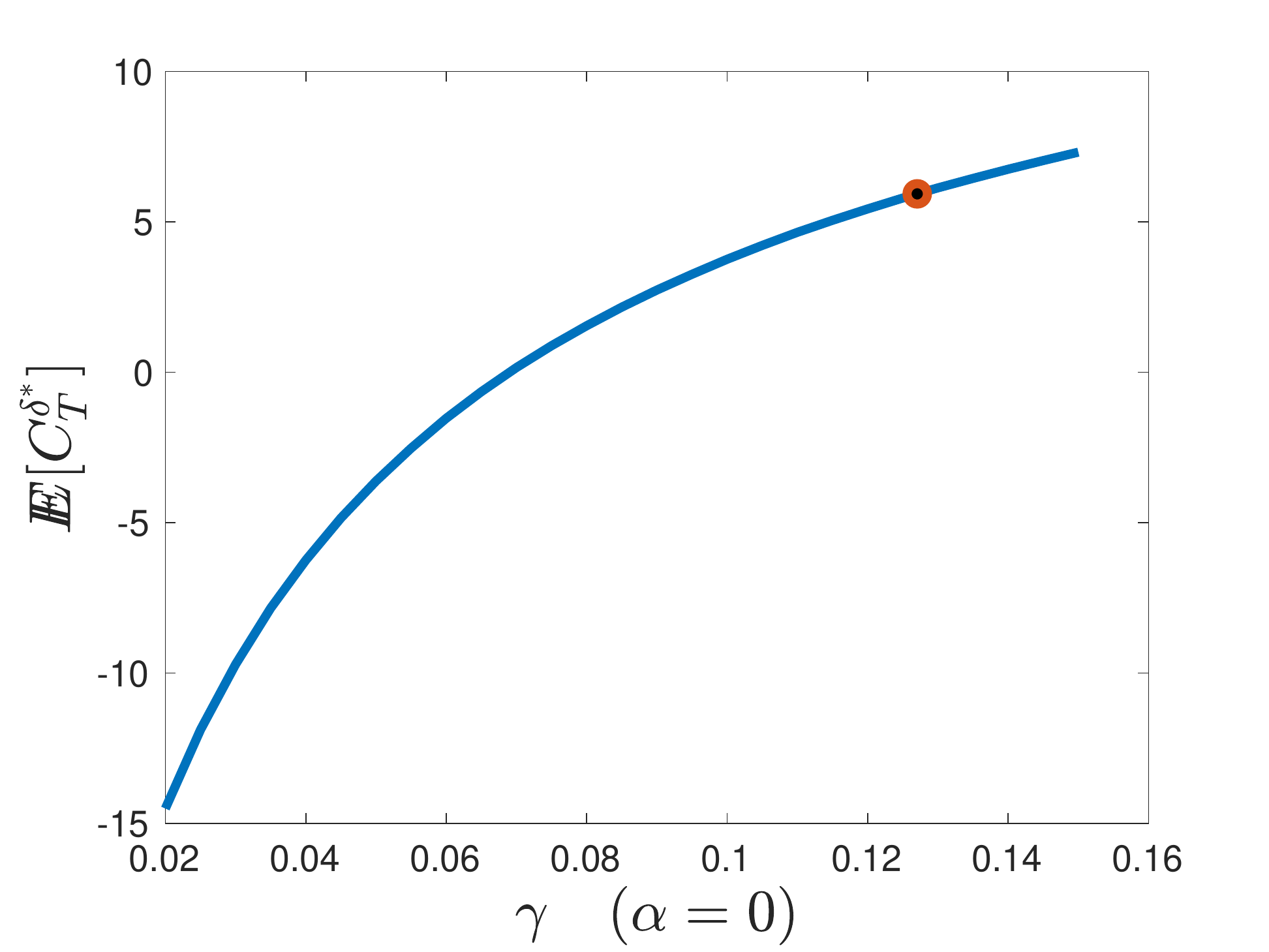}
\caption{Top panel  shows $\mathbb{P} (D^{\delta^*}_T<0.1\,N_T )$ and  $\mathbb{E} [C^{\delta^*}_T ]$ when $\gamma=0$ and for $\alpha\in[0,2.5]$, recall that $\delta^*=\alpha$ when $\gamma=0$, see \eqref{simple control gamma zero}. Similarly, bottom panel  shows $\mathbb{P} (D^{\delta^*}_T<0.1\,N_T )$ and  $\mathbb{E}[C^{\delta^*}_T]$ when $\alpha=0$  and $\gamma\in[0.02,0.16]$. In all pictures, the orange circle  marks the lowest value of $\mathbb{E}\left[C_T\right]$ when $\mathbb{P}\left(D_T<0.1\,N_T\right)\geq 0.99$. Other model parameters: $\lambda=100$, $\alpha=0$, and $Z\sim \mathbf N(0.2,1)$ for all trades.}
\label{Fig:Risk Metric}
\end{center}
\end{figure}

\subsection{Pinned arrival rates}

In this section, we assume  the arrival intensity of the agent's MLOs is
\begin{equation}\label{eqn:lambda star epsilon}
\lambda^{\star}_t=\frac{M-N_{t^-}}{T-t+\epsilon}\,,
\end{equation}
where $M>0$ is a positive integer,   $\epsilon>0$ and recall that $N_t$ denotes the number of trade attempts. The intensity $\lambda^{\star}_t$  is bounded by $\bar{\lambda}=M/\epsilon$,  which is a condition we require in the latency-optimal strategy we derived above, and if $\epsilon=0$, the intensity guarantees that $ N_T=M$, see \cite{conforti2016bridges} and \cite{hoyle2010}.

Now, use the Markov property of $\delta^*$ to write $\delta^*=h(t,D_{t^-},N_{t^-})$, where the function $h$ satisfies the PIDE
\begin{align*}
0=\partial_t h(t,D,N)&+\left(\int_{h(t,D,N)}^{\infty} \frac{M-N}{T-t+\epsilon}\,\phi_t(z)\,\diff z\right)\,\left(h(t,D+1,N+1)-h(t,D,N)\right)\\
 &+\left(\int_{-\infty}^{h(t,D,N)} \,\frac{M-N}{T-t+\epsilon}\, \phi_t(z)\,\diff z\right)\,\left(h(t,D,N+1)-h(t,D,N)\right)\,,
\end{align*}
with
\begin{align*}
h(t,D,M)=2\,\gamma\,D+\gamma+\alpha \qquad\text{and}\qquad h(T,D,N)=2\,\gamma\,D+\gamma+\alpha\,.
\end{align*}

Figure \ref{Fig:DeltaStar h function pinned}  shows the optimal discretion to walk the LOB for various values of missed trades and target number of trades $M=100$. The interpretation is similar to that of Figure \ref{Fig:DeltaStar h function}.

\begin{figure}[H]
\begin{center}
\includegraphics[width=0.32\textwidth]{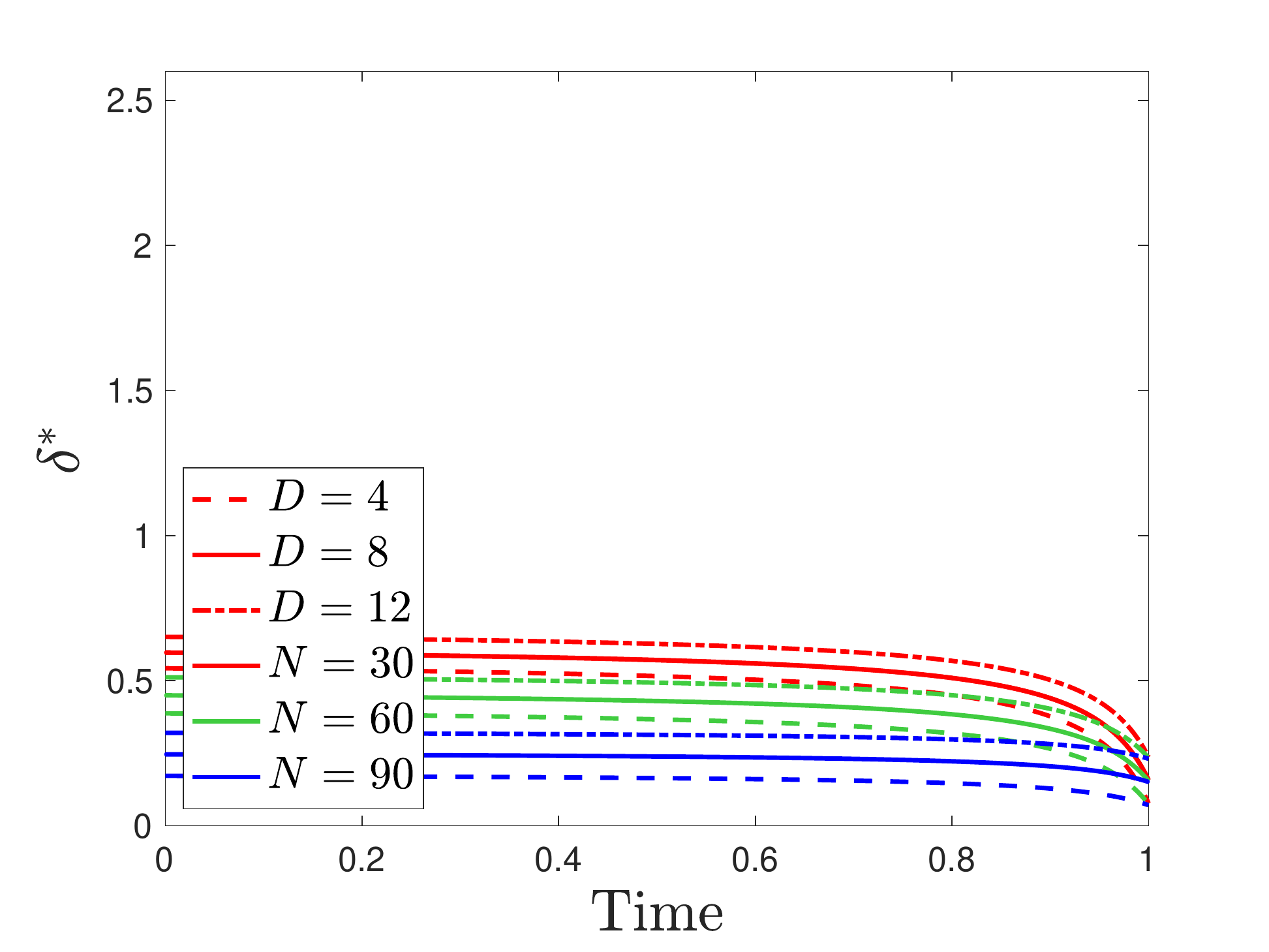}
\includegraphics[width=0.32\textwidth]{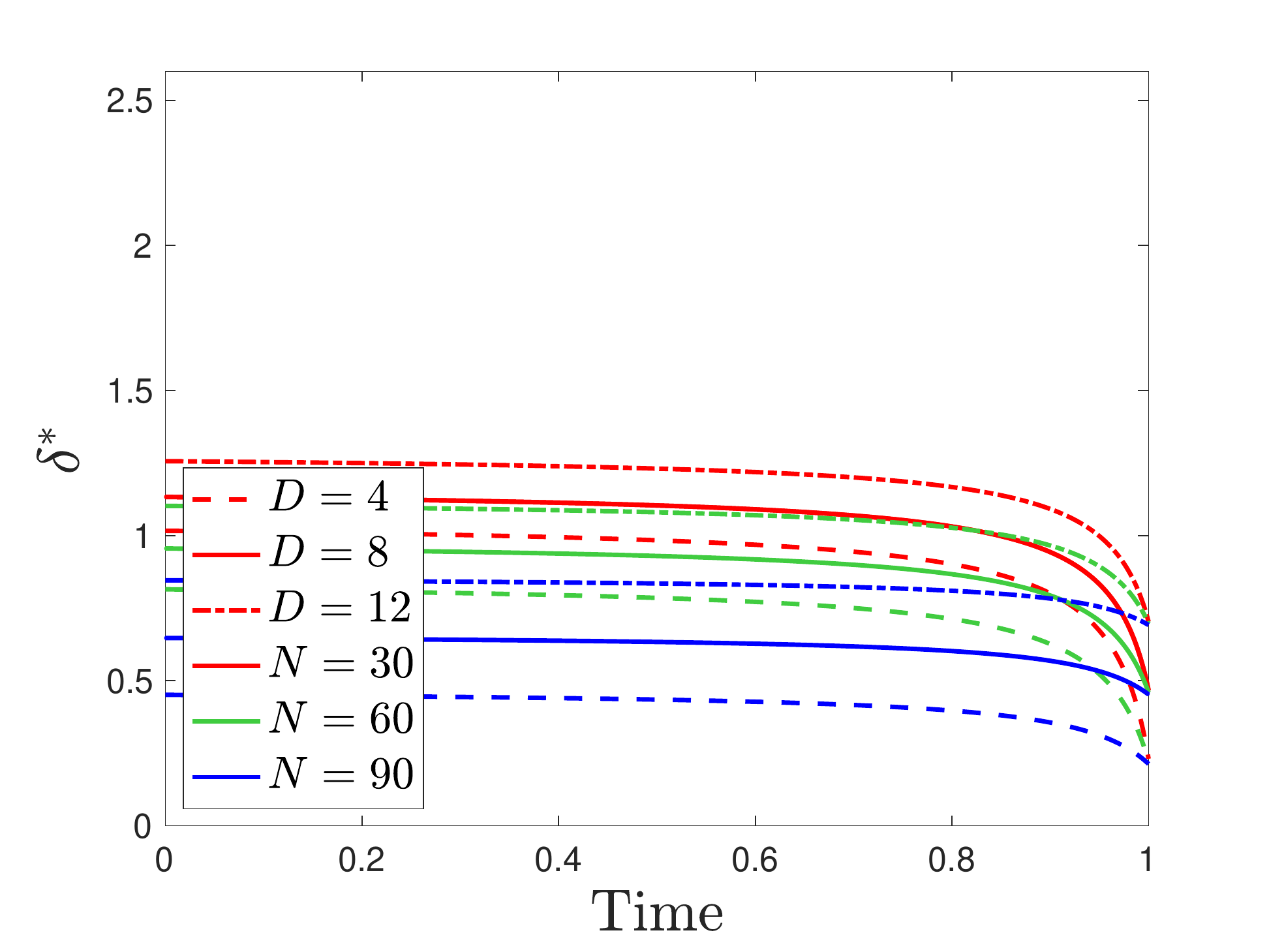}
\includegraphics[width=0.32\textwidth]{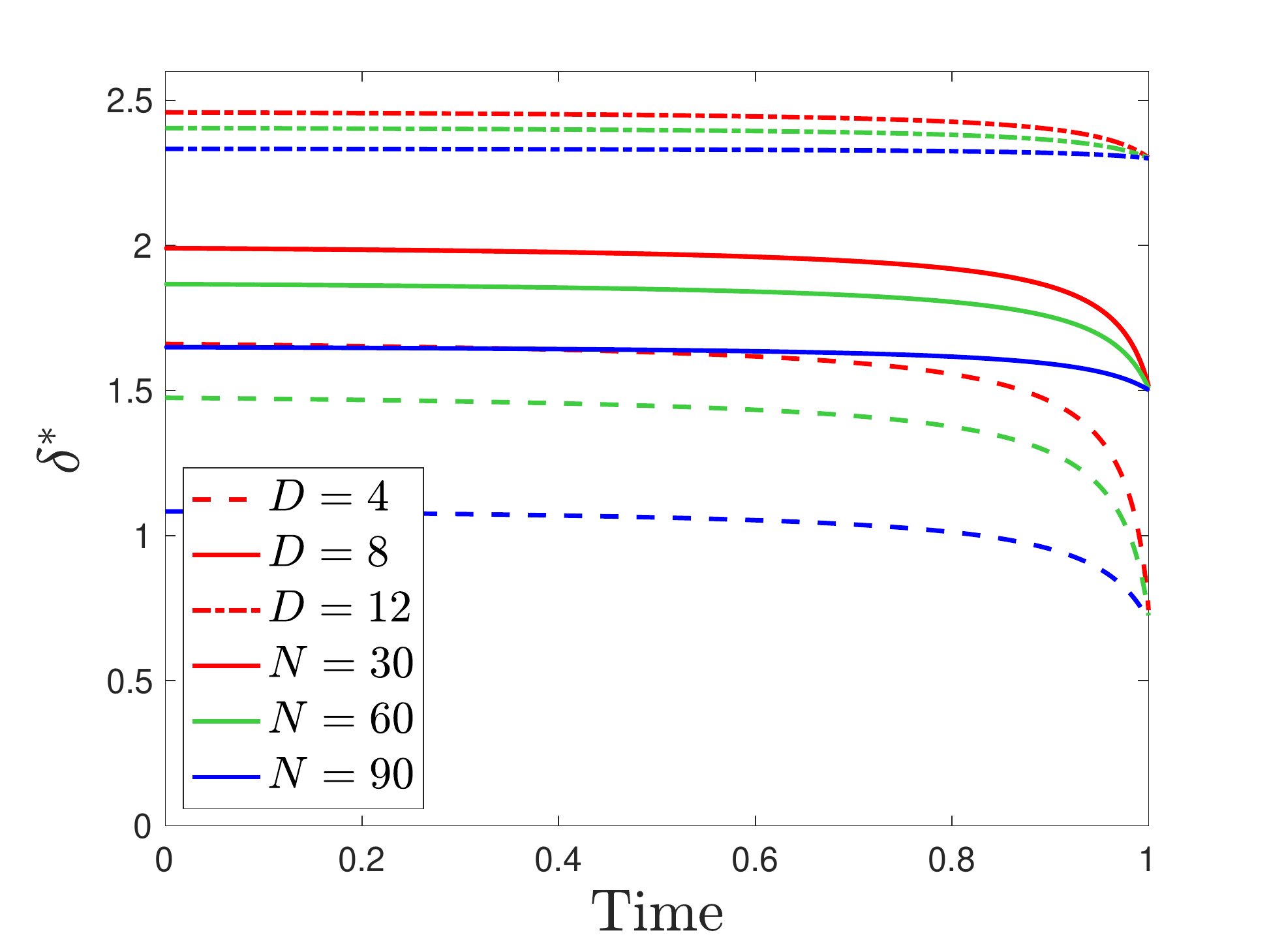}
\caption{Optimal strategy $\delta^*$ for various values of $\gamma$, number of misses, and number of attempts. From left to right, penalty parameter is $\gamma=0.01$, $\gamma=0.03$, and $\gamma=0.1$. Dotted line  $N_t=30$, solid line $N_t=60$, and  dot-dash line $N_t=90$. Blue lines $D_t^{\delta^*}=4$, green lines  $D_t^{\delta^*}=8$, red lines $D_t^{\delta^*}=12$. The remaining parameters are: $M=100$, $\alpha=0$, $\epsilon=0.1$, $Z\sim \mathbf N(0.2,1)$.}
\label{Fig:DeltaStar h function pinned}
\end{center}
\end{figure}

We perform 10,000 simulations with the same parameters as above and use the arrival rate of the  MLOs as in \eqref{eqn:lambda star epsilon} with $\epsilon = 0.1$. Figures   \ref{Fig:Sample Paths Pinned} and  \ref{Fig:Histograms_DiffGammas_Pinned} report the results, which have a similar interpretation to that of Figures  \ref{Fig:Paths_Simulation}   and \ref{Fig:Histograms_DiffGammas}, respectively.

\begin{figure}[H]
\centering
\includegraphics[width=0.75\textwidth]{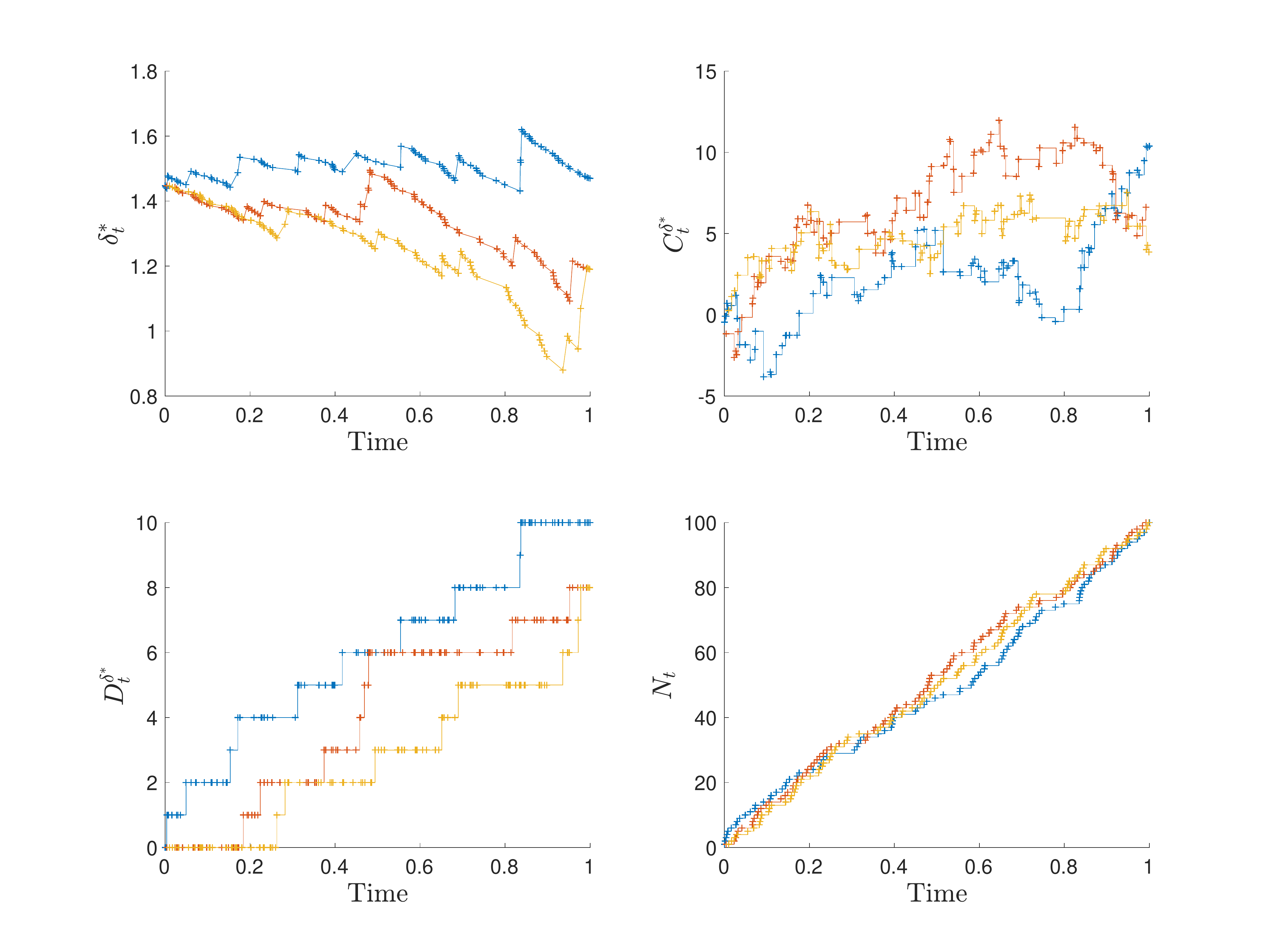}
\caption{Sample paths for the optimal discretion $\delta^*$ (top left panel), number of missed trades $D^{\delta^*}$ (lower left panel), cost of strategy $C^{\delta^*}$ (top right panel), and number of trade attempts $N$ (lower right panel) for three simulations of the MPP. Parameters: $\alpha=0$, $\gamma = 0.07$, $\epsilon=0.1$, $M =100$, $T=1$. }\label{Fig:Sample Paths Pinned}
\end{figure}

\begin{figure}[H]
\begin{center}
\includegraphics[width=0.36\textwidth]{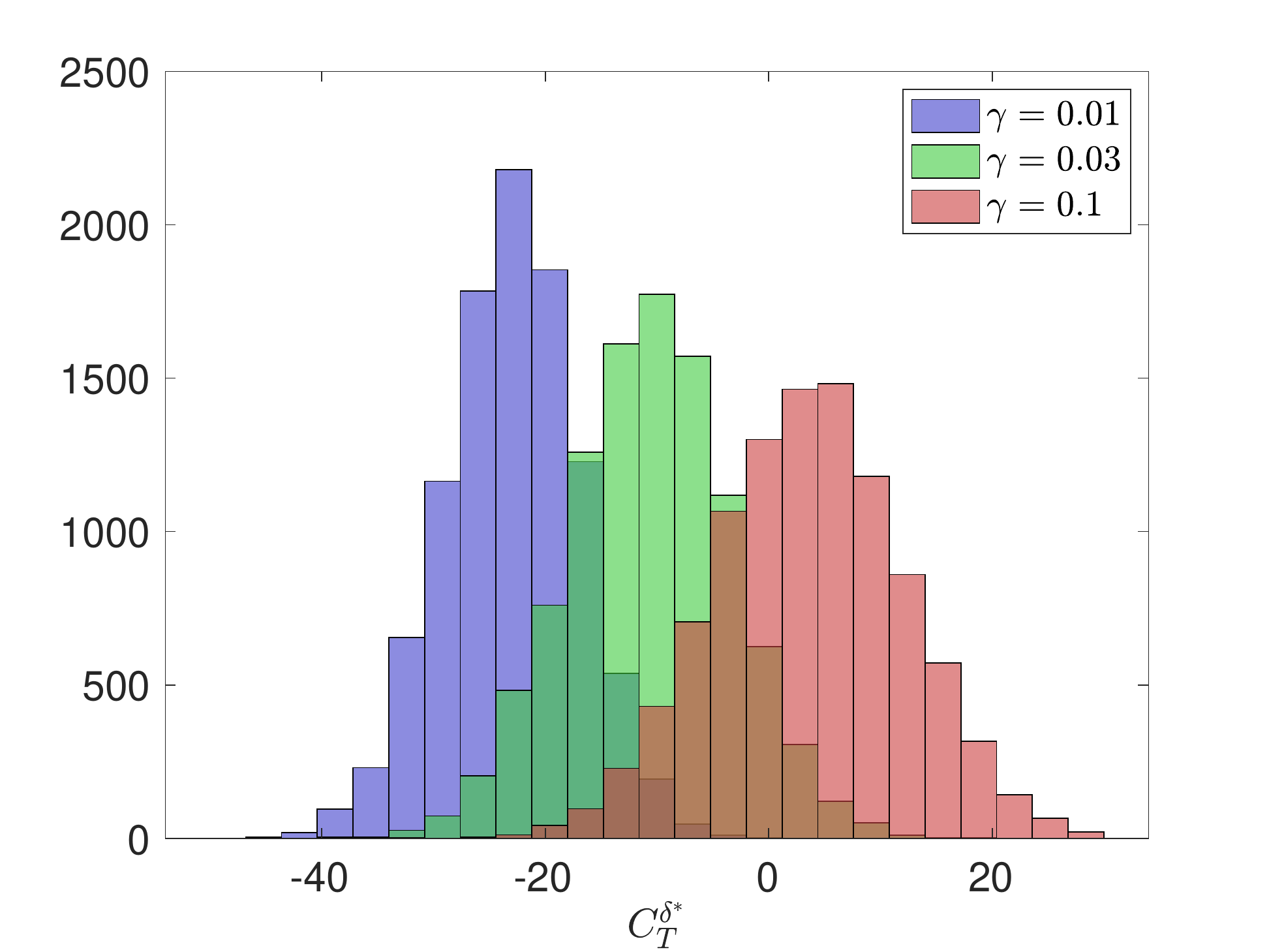}
\includegraphics[width=0.36\textwidth]{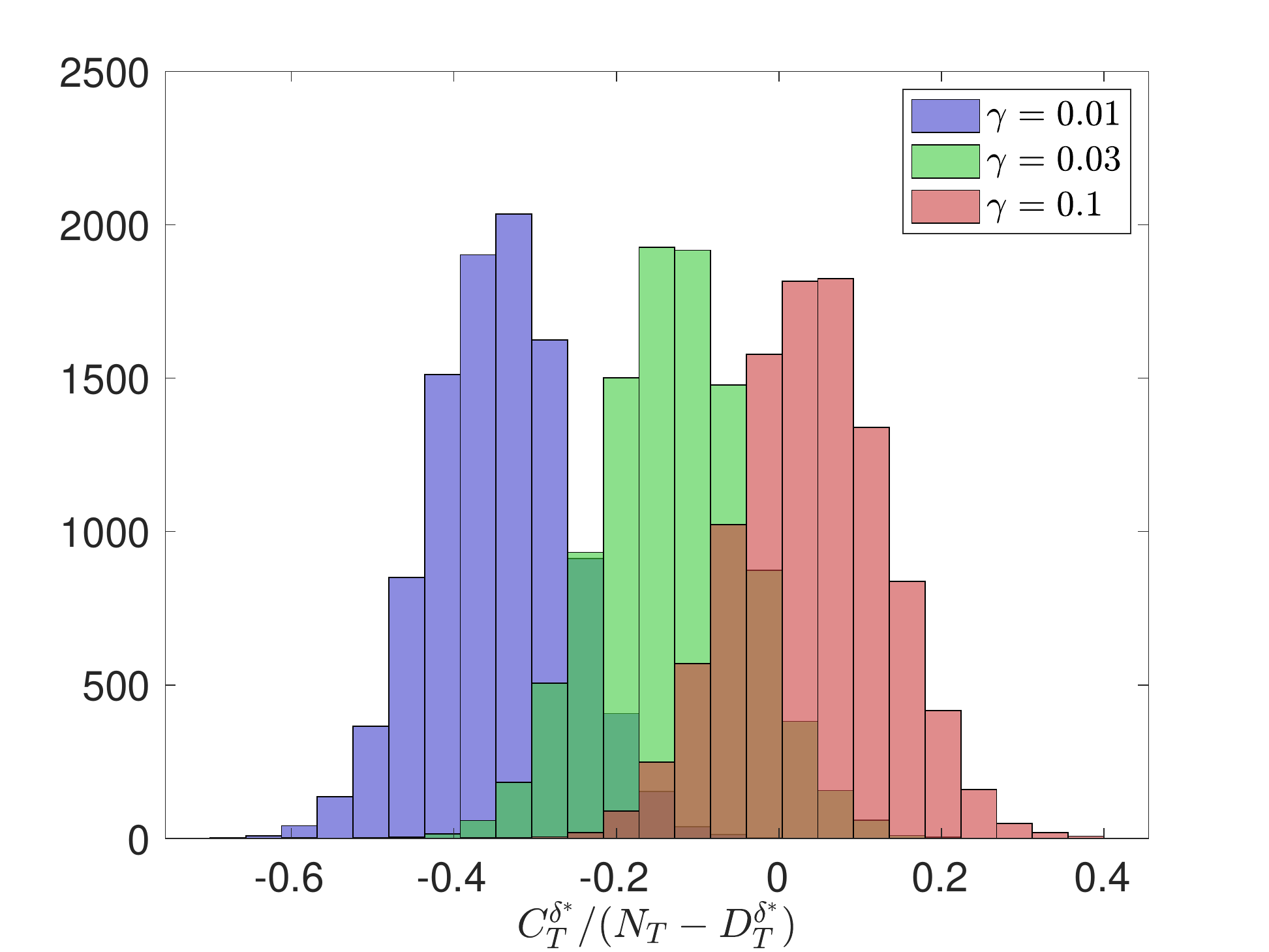}
\includegraphics[width=0.36\textwidth]{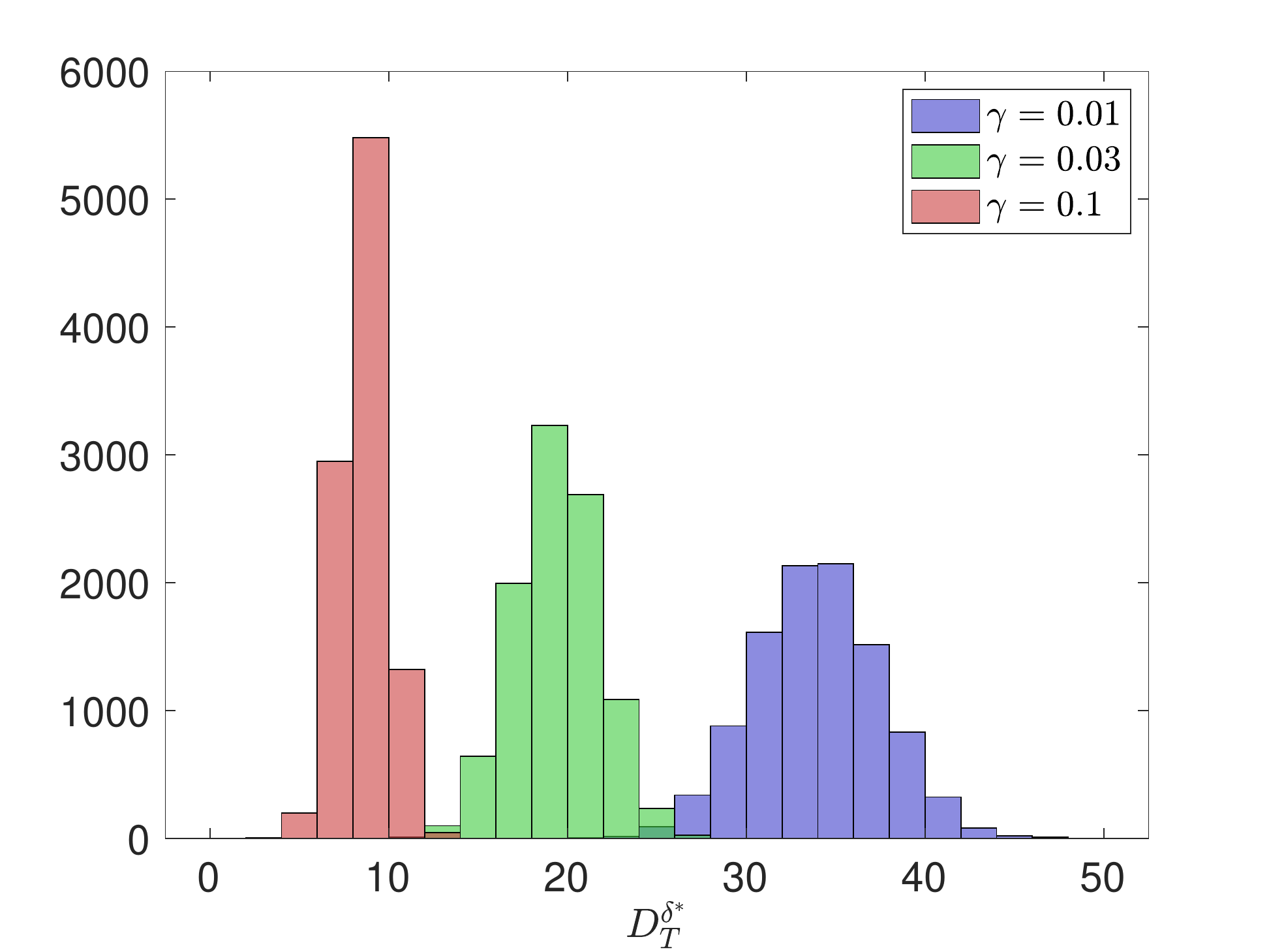}
\includegraphics[width=0.36\textwidth]{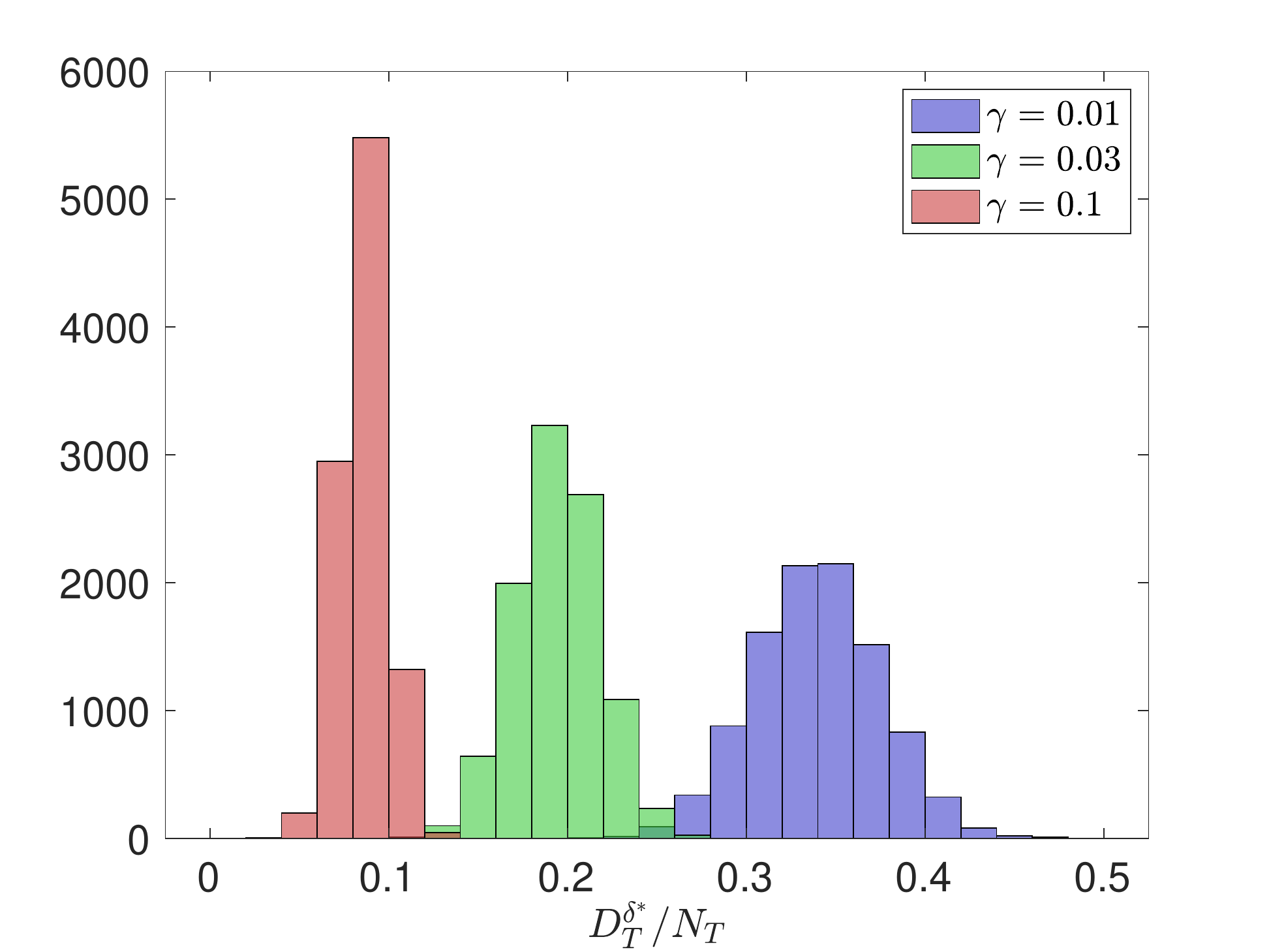}
\caption{Top left panel: Histogram of the cost $C^{\delta^*}_T$ of the strategy. Top right panel: Histogram of the extra cost per filled trade $C^{\delta^*}_T /  (N_T - D^{\delta^*}_T)$. Bottom left panel: Histogram of the number of misses $D^{\delta^*}_T$. Bottom right panel: Histogram of percentage of misses $D^{\delta^*}_T/N_T$. }
\label{Fig:Histograms_DiffGammas_Pinned}
\end{center}
\end{figure}

\section{Conclusions}\label{sec:conclusions}

With few exceptions, the literature on algorithmic trading assumes that latency in the marketplace is zero. This is not accurate, and the effects of latency on the efficacy of liquidity making and taking strategies are economically significant. In this paper we proposed a model to improve the marksmanship of the orders sent by liquidity takers when, due to latency, the limit order book is a moving target. 

We showed how a liquidity taker chooses the price limit  of marketable orders when there is latency in the marketplace. The optimal strategy balances the tradeoff between the costs of walking the book and the number of missed trades over a trading horizon. We modelled the effects of latency as a marked point process that captures the interaction between liquidity taking orders and the limit orders resting in the book. We characterized the optimal price limit of marketable orders as a solution to a FBSDE, which, to the best of our knowledge, is new and, as the extant literature does not have  uniqueness and existence results, we prove both.

The strategy developed here may be implemented as another layer  of any liquidity taking strategy (especially those that follow a stochastic trading schedule) that incorrectly assumes zero latency. Our framework can be applied in other contexts too. In its most general form, we solve a problem in which the agent decides how much she is willing to pay to absorb a stochastic shock to achieve an objective or complete a task. For example, market makers in  foreign exchange markets  with `last look' can employ the framework developed in this paper. The last look feature allows liquidity makers to reject trades, so they are not picked off by faster liquidity taking traders, see \cite{Oomen17} and \cite{Cartea15LL}. Specifically, with our framework,  a foreign exchange market maker can obtain the optimal tolerance that maximizes the number of incoming marketable orders she is willing to fill  while minimizing losses to the fast traders who snipe her stale quotes in the LOB.

\bibliographystyle{apalike}

\bibliography{References}

\newpage
\let\normalsize\small
\appendix
\small
\input{Appendix.tex}

\end{document}

%% file: Appendix.tex
\section{Proof of Lemma \ref{Gat:CostLPandQP} }\label{proof of lemma}
We prove the lemma  in three parts. First we work out the {G\^ateaux derivative of the cost function}. We use \eqref{eqn:jc} to write
\begin{align*}
\frac{1}{\epsilon}\,\left\{\mathbb{E}\left[C^{\delta+\epsilon\,w}_T\right]-\mathbb{E}\left[C^{\delta}_T\right]\right\}&=\frac{1}{\epsilon}\,\left\{\mathbb{E}\left[\int_0^T \int_{\mathbb{R}} z\,\left(\Ffct(\delta_t+\epsilon\,w_t-z)-\Ffct(\delta_t-z)\right)\,\phi_t(\diff z)\,\diff A_t\right]\right\}\\
&=\frac{1}{\epsilon}\,\left\{\mathbb{E}\left[\int_0^T \int_{\delta_t}^{\delta_t+\epsilon\,w_t} z\,\phi_t(\diff z)\,\diff A_t\right]\right\}\\
&=\mathbb{E}\left[\int_0^T \frac{1}{\epsilon}\left\{ \int_{\delta_t}^{\delta_t+\epsilon\,w_t} z\,\phi_t(\diff z)\right\}\,\diff A_t\right]\,.
\end{align*}
Then, by the dominated convergence theorem and the fundamental theorem of calculus, we have
\begin{align*}
\langle \mathcal{D}\,J^{\text{C}}(\delta),w\rangle&=\lim_{\epsilon\to 0}\frac{J^{\text{C}}(\delta+\epsilon \,w)-J^{\text{C}}(\delta)}{\epsilon}\\
&=\lim_{\epsilon\to 0}\frac{1}{\epsilon}\,\left\{\mathbb{E}\left[C^{\delta+\epsilon\,w}_T\right]-\mathbb{E}\left[C^{\delta}_T\right]\right\}=\mathbb{E}\left[\int_0^T \delta_t\,w_t\,\phi_t(\delta_t)\,\diff A_t\right]\,.
\end{align*}

\noindent Next we work out the {G\^ateaux derivative of the linear penalty}. Note that
\begin{align*}
\frac{1}{\epsilon}\,\left\{\mathbb{E}\left[D^{\delta+\epsilon\,w}_T\right]-\mathbb{E}\left[D^{\delta}_T\right]\right\}&=  \frac{1}{\epsilon}\,\left\{\mathbb{E}\left[\int_0^T\int_{\mathbb{R}}\left(G(\delta_t+\epsilon\,w_t-z)-G(\delta_t-z)\right)\,\rcompr(\diff z,\,\diff t)\right]\right\}\\
&=\mathbb{E}\left[\int_0^T \frac{1}{\epsilon}\,\left\{ \int^{\delta_t}_{\delta_t+\epsilon\,w_t} \phi_t(\diff z)\right\}\,\diff A_t\right]\,,
\end{align*}
therefore, we have
\begin{align*}
\lim_{\epsilon\to 0}\frac{1}{\epsilon}\,\bigg\{\mathbb{E}\left[D^{\delta+\epsilon\,w}_T\right]-\mathbb{E}\left[D^{\delta}_T\right]\bigg\}&=-\mathbb{E}\left[\int_0^T w_t\,\phi_t(\delta_t)\,\diff A_t\right]\,.
\end{align*}

\noindent Finally, we work out the {G\^ateaux derivative of the quadratic penalty}. We write 
\begin{align*}
\frac{1}{\epsilon}\,\Bigg\{\mathbb{E}\left[\left(D^{\delta+\epsilon\,w}_T\right)^2\right]-\mathbb{E}\left[\left(D^{\delta}_T\right)^2\right]\Bigg\}&= \frac{2}{\epsilon}\,\Bigg\{\mathbb{E}\left[\int_0^T\int_{\mathbb{R}}D^{\delta+\epsilon\,w}_{t^-}\,G(\delta_t+\epsilon\,w_t-z)\,\rcompr(\diff z,\,\diff t)\right]  \\
&\quad\quad\quad-\mathbb{E}\left[\int_0^T\int_{\mathbb{R}}D^{\delta}_{t^-}\,G(\delta_t-z)\,\rcompr(\diff z,\,\diff t)\right]\Bigg\}\\
&\quad+ \frac{1}{\epsilon}\,\Bigg\{\mathbb{E}\left[\int_0^T\int_{\mathbb{R}}G(\delta_t+\epsilon\,w_t-z)-G(\delta_t-z)\,\rcompr(\diff z,\,\diff t)\right]\Bigg\}\,.
\end{align*}
Subtract and add
\begin{equation*}
\frac{2}{\epsilon}\,\bigg\{\mathbb{E}\left[\int_0^T\int_{\mathbb{R}}D^{\delta+\epsilon\,w}_{t^-}\,G(\delta_t-z)\,\rcompr(\diff z,\,\diff t)\right]\bigg\}
\end{equation*}
to the right-hand side of the equation above and write
\begin{align*}
&\frac{1}{\epsilon}\,\Bigg\{\mathbb{E}\left[\left(D^{\delta+\epsilon\,w}_T\right)^2\right]-\mathbb{E}\left[\left(D^{\delta}_T\right)^2\right]\Bigg\}
\\
&\quad= \frac{2}{\epsilon}\,\Bigg\{\mathbb{E}\left[\int_0^T\int_{\mathbb{R}} \left(D^{\delta+\epsilon\,w}_{t^-}\,G(\delta_t+\epsilon\,w_t-z)-D^{\delta+\epsilon\,w}_{t^-}\,G(\delta_t-z)\right)\,\rcompr(\diff z,\,\diff t)\right]  \\
&\quad\quad\quad\quad+\mathbb{E}\left[\int_0^T\int_{\mathbb{R}} \left(D^{\delta+\epsilon\,w}_{t^-}\,G(\delta_t-z)-D^{\delta}_{t^-}\,G(\delta_t-z)\right)\,\rcompr(\diff z,\,\diff t)\right]\Bigg\}\\
&\quad\quad+ \frac{1}{\epsilon}\,\Bigg\{\mathbb{E}\left[\int_0^T\int_{\mathbb{R}} \left(G(\delta_t+\epsilon\,w_t-z)-G(\delta_t-z)\right)\,\rcompr(\diff z,\,\diff t)\right]\Bigg\}\\
&\quad= 2\,\Bigg\{{\mathbb{E}\left[\int_0^T\int_{\mathbb{R}}\left(D^{\delta+\epsilon\,w}_{t^-}\,\frac{G(\delta_t+\epsilon\,w_t-z)-G(\delta_t-z)}{\epsilon}\right)\,\rcompr(\diff z,\,\diff t)\right]}\tag{QP1}  \\
&\quad\quad\quad\quad+{\mathbb{E}\left[\int_0^T\int_{\mathbb{R}}\left(\frac{D^{\delta+\epsilon\,w}_{t^-}-D^{\delta}_{t^-}}{\epsilon}\,G(\delta_t-z)\right)\,\rcompr(\diff z,\,\diff t)\right]}\tag{QP2}\Bigg\}\\
&\quad\quad+ {\mathbb{E}\left[\int_0^T\int_{\mathbb{R}}\left(\frac{G(\delta_t+\epsilon\,w_t-z)-G(\delta_t-z)}{\epsilon}\right)\,\rcompr(\diff z,\,\diff t)\right]}\tag{QP3}\,.
\end{align*}

Next, take the limit of  QP1, QP2, and QP3 as $\epsilon$ approaches zero.  The limit of QP1 is given by 
\begin{align*}
\lim_{\epsilon\to 0}\text{QP1}&=\lim_{\epsilon\to 0}\mathbb{E}\left[\int_0^T\int_{\mathbb{R}}D^{\delta+\epsilon\,w}_{t^-}\,\frac{G(\delta_t+\epsilon\,w_t-z)-G(\delta_t-z)}{\epsilon}\,\rcompr(\diff z,\,\diff t)\right]\\
&=\lim_{\epsilon\to 0}\mathbb{E}\left[\int_0^T\,D^{\delta+\epsilon\,w}_{t^-}\,\frac{1}{\epsilon}\,\left\{\int_{\delta_t+\epsilon\,w_t}^{\delta_t} \phi_t(\diff z)\,\right\}\diff A_t\right]\\
&=-\mathbb{E}\left[\int_0^T D^{\delta}_{t^-}\,w_t\,\phi_t(\delta_t)\,\diff A_t\right]\,.
\end{align*}
The last equality follows from the dominated convergence theorem and because $\lim_{\epsilon\to 0}D^{\delta+\epsilon\,w}_{t^-}=D^{\delta}_{t^-}$ almost surely.

\noindent The limit of QP2 is given by 
\begin{align*}
&\lim_{\epsilon\to 0}\text{QP2}
\\
&\quad=\lim_{\epsilon\to 0}\mathbb{E}\left[\int_0^T\int_{\mathbb{R}}\frac{D^{\delta+\epsilon\,w}_{t^-}-D^{\delta}_{t^-}}{\epsilon}\,G(\delta_t-z)\,\rcompr(\diff z,\,\diff t)\right]\\
&\quad=\lim_{\epsilon\to 0}\mathbb{E}\left[\int_0^T\int_{\mathbb{R}}G(\delta_t-z)\,\left(\int_0^{t^-}\int_{\mathbb{R}}\frac{G(\delta_s+\epsilon\,w_s-z')-G(\delta_s-z')}{\epsilon}\,\rmes(\diff z',\,\diff s)\right)\,\rcompr(\diff z,\,\diff t)\right]\\
&\quad=\lim_{\epsilon\to 0}\mathbb{E}\left[\int_0^T\int_{\mathbb{R}}\frac{G(\delta_t+\epsilon\,w_t-z)-G(\delta_t-z)}{\epsilon}\,\left(\int_{t}^{T}\int_{\mathbb{R}}G(\delta_s-z)\,\rcompr(\diff z',\,\diff s)\right)\,\rmes(\diff z,\,\diff t)\right]\\
&\quad=\lim_{\epsilon\to 0}\mathbb{E}\left[\int_0^T\int_{\mathbb{R}}\frac{G(\delta_t+\epsilon\,w_t-z)-G(\delta_t-z)}{\epsilon}\,\mathbb{E}_{t^-}\left[\int_{t}^{T}\int_{\mathbb{R}}G(\delta_s-z)\,\rcompr(\diff z',\,\diff s)\right]\,\rmes(\diff z,\,\diff t)\right]\\
&\quad=\lim_{\epsilon\to 0}\mathbb{E}\left[\int_0^T\int_{\mathbb{R}}\frac{G(\delta_t+\epsilon\,w_t-z)-G(\delta_t-z)}{\epsilon}\,\mathbb{E}_{t^-}\left[\int_{t}^{T}\int_{\mathbb{R}}G(\delta_s-z)\,\rcompr(\diff z',\,\diff s)\right]\,\rcompr(\diff z,\,\diff t)\right]\\
&\quad=\lim_{\epsilon\to 0}\mathbb{E}\left[\int_0^T\,\frac{1}{\epsilon}\,\int_{\delta_t+\epsilon\,w_t}^{\delta_t}\,\phi_t(\diff z)\,\mathbb{E}_{t^-}\left[\int_{t}^{T}\int_{\mathbb{R}}G(\delta_s-z)\,\rcompr(\diff z',\,\diff s)\right]\,\diff A_t\right]\\
&\quad=-\mathbb{E}\left[\int_0^T\,w_t\,\phi_t(\delta_t)\,\mathbb{E}_{t^-}\left[\int_{t}^{T}\int_{\mathbb{R}}G(\delta_s-z)\,\rcompr(\diff z',\,\diff s)\right]\,\diff A_t\right]\,.
\end{align*}

Finally, the limit of QP3 is given by 
\begin{align*}
\lim_{\epsilon\to 0}\text{QP3}&=\lim_{\epsilon\to 0}\mathbb{E}\left[\int_0^T\int_{\mathbb{R}}\frac{G(\delta_t+\epsilon\,w_t-z)-G(\delta_t-z)}{\epsilon}\,\rcompr(\diff z,\,\diff t)\right]\\
&=\lim_{\epsilon\to 0}\mathbb{E}\left[\int_0^T\,\frac{1}{\epsilon}\int_{\delta_t+\epsilon\,w_t}^{\delta_t}\phi_t(\diff z)\,\diff A_t\right]\\
&=-\mathbb{E}\left[\int_0^T w_t\,\phi_t(\delta_t)\,\diff A_t\right]\,,
\end{align*}
which concludes the proof.

\section{Bounded G\^ateaux derivative}\label{BddGD}
Let $S=\max\{1,\,\{\phi_t(z)\}_{0\leq t \leq T \,, z\in \mathbb{R}}\}<\infty$ and   $\delta,\,w\,\in\,\mcA$.  Let  $\eta_t=\max\{\delta_t,\,w_t,\,N_{t^-}\}$,  which is predictable because each process is predictable,  and note that $\mathbb{E}\left[\sup_{0\leq t\leq T}(\eta_t)^2\right]\leq 4\,\mathbb{E}\left[\sup_{0\leq t\leq T}(\delta_t)^2\right]+4\,\mathbb{E}\left[\sup_{0\leq t\leq T}(w_t)^2\right]+4\,\mathbb{E}\left[\sup_{0\leq t\leq T}(N_{t^-})^2\right]<\infty$. Then 
\begin{eqnarray*}
\abs{\langle \mathcal{D}\,J(\delta),w\rangle}&\leq& \abs{\mathbb{E}\left[\int_0^T  \delta_t\,w_t\,\phi_t(\delta_t)\,\diff A_t\right]}+2\,\gamma\,\abs{\mathbb{E}\left[\int_0^T \phi_t(\delta_t)\,w_t\,\left(\int_t^T\int_{\mathbb{R}}G(\delta_s-z')\rcompr(\diff z',\,\diff s)\right)  \diff A_t\right]}\\
&& 2\,\gamma\,\abs{\mathbb{E}\left[\int_0^T \phi_t(\delta_t)\,w_t\,D^{\delta}_{t^-}\, \diff A_t\right]}+(\gamma+\alpha)\,\abs{\mathbb{E}\left[\int_0^T  \phi_t(\delta_t)\,w_t\,\diff A_t\right]}\\
&\leq &S\,\bar{\lambda}\,T\,\mathbb{E}\left[\sup_{0\leq t\leq T}(\eta_t)^2\right]+2\,\gamma\,S\,T^2\,\bar{\lambda}\mathbb{E}\left[\sup_{0\leq t\leq T}\abs{\eta_t}\right]\\
&&+2\,\gamma\,S\,T\,\bar{\lambda}\,\mathbb{E}\left[\sup_{0\leq t\leq T}(\eta_t)^2\right]+(\gamma+\lambda)\,S\,T\,\mathbb{E}\left[\sup_{0\leq t\leq T}\abs{\eta_t}\right]\\
&<&\infty\,.
\end{eqnarray*}

\section{Second G\^ateaux derivative}\label{SecondGD}
The first G\^ateaux derivative of the functional $J$ is given by
\begin{align*}
{\langle \mathcal{D}\,J(\delta),w\rangle}&= \mathbb{E}\left[\int_0^T  w_t\,\phi_t(\delta_t)\,\left(\delta_t-2\,\gamma\,\left(\int_t^T\int_{\mathbb{R}}G(\delta_s-z')\,\rcompr(\diff z',\,\diff s)\right)-2\,\gamma\,D^{\delta}_{t^-}-(\gamma+\alpha) \right)\,\diff A_t\right]\\
&= \mathbb{E}\left[\int_0^T  w_t\,\phi_t(\delta_t)\,\left(\delta_t-2\,\gamma\,\mathbb{E}_{t^-}\left[D^{\delta}_{T}\right]-(\gamma+\alpha) \right)\,\diff A_t\right]\,.
\end{align*}
Let $\delta,\,w,\,\nu\in\mcA$. The second G\^ateaux derivative of $J(\delta)$ in the directions $w$ and $\nu$, is defined as
\begin{align*}
{\langle \mathcal{D}^2\,J(\delta),w,\nu\rangle}&= \lim_{\epsilon\to 0}\frac{{\langle \mathcal{D}\,J(\delta+\epsilon\,\nu),w\rangle}-{\langle \mathcal{D}\,J(\delta),w\rangle}}{\epsilon}\,,
\end{align*}
which converges to
\begin{eqnarray*}
{\langle \mathcal{D}^2\,J(\delta),w,\nu\rangle}&=& \mathbb{E}\left[\int_0^T w_t\,\nu_t \,\phi_t'(\delta_t)\,\left(\delta_t-2\,\gamma\,\mathbb{E}_{t^-}\left[D^{\delta}_T\right]-\gamma-\alpha\right)\,\diff A_t\right]\\
&& +\mathbb{E}\left[\int_0^T w_t\,\phi_t(\delta_t)\,\left(\nu_t+2\,\gamma\,\mathbb{E}_{t^-}\left[\int_0^T \phi_s(\delta_s)\,\nu_s\,\diff A_s\right]\right)\,\diff A_t\right]\,.
\end{eqnarray*}